\documentclass[11pt]{amsart}
\usepackage{personal-commands}

\begin{document}
	
\begin{spacing}{1.2}

\title[Resonances with vanishing interactions at crossings]{Semiclassical resonances for matrix Schrödinger operators with vanishing interactions at crossings of classical trajectories}

\author{Vincent Louatron}


\begin{abstract}
	We study the semiclassical distribution of resonances of a $2 \times 2$ matrix Schrödinger operator \eqref{Eq:MatrixSchrödingerOperator} near a fixed energy level where the underlying classical trajectories $\Gamma_1$ of $P_1$ and $\Gamma_2$ of $P_2$ are respectively periodic and non-trapping. The aim is to compute the imaginary part of the resonances appearing near the eigenvalues created by $P_1$ when $\Gamma_1$ intersects with $\Gamma_2$ with finite contact order $m$.
	Recent results \cite{AFH} and \cite{AFH2} assert that the width of resonances is of polynomial order in the semiclassical parameter with exponent $1+2/(m+1)$, if the interaction $U=r_0(x) + ihr_1(x)D_x$ is elliptic at the crossing point. Here, we remove this ellipticity assumption and prove that the exponent is $1+2(k+1)/(m+1)$ where $k$ is a "vanishing order" of the interaction at the crossing points, suitably defined in terms of the vanishing orders of $r_0$ and $r_1$ depending on whether the crossing point is a turning point or not. 
\end{abstract}

\maketitle

\section{Introduction}
In this paper, we look at the asymptotic distribution of quantum resonances in the semiclassical limit for a matrix Schrödinger operator. More specifically, we are concerned with the influence that the crossings of the classical trajectories associated with this operator have on this distribution. Our motivation comes from molecular predissociation problems that have been studied under the Born-Oppenheimer approximation for many years for scalar operators and more recently for matrix-valued operators (see for example references in \cite{KMSW} for some history on the topic).
In this approximation, the problem can be reduced to the study of the matrix Schrödinger operator
\begin{equation}\label{Eq:MatrixSchrödingerOperator} P := 
	\begin{pmatrix}
		P_1 & h U\\
		h U^* & P_2
	\end{pmatrix}
\end{equation}
\noindent
with 
\[
P_j :=(h D_x)^2  + V_j(x),\, j \in \{1,2\}, \stext{ and} D_x:= -i \partial_x,
\]
where $h>0$ denotes the semiclassical parameter ($h^2$ is the ratio of the electronic mass over the nuclear mass) and $U = r_0(x) + ir_1(x)hD_x$ a first order (semiclassical) differential operator. 

We provide a generalization of the main results of \cite{AFH} and \cite{AFH2}, which give the distribution of semiclassical resonances of $P$ near an energy $E_0$ such that $V_1$ has a simple well on one hand, and $E_0$ is a non-trapping energy for $V_2$ on the other hand (see figures \ref{Fig:PotentialsCrossings} and \ref{Fig:PotentialsCrossingsTypeB}). 
Such distribution was first studied in \cite{FMW3} for an energy level above the crossing level (referred to as cases \ref{Case:Ia} and \ref{Case:Ib} below) for transversal crossings ($m=1$ below). In this work, under some elliptic condition at $x=0$ on the interaction term $U$, the problem is brought back to the study of first order differential equations via normal forms
. Due to the important role of the ellipticity assumption in this normal form method, it seems hard to deal with the non-elliptic case. However, authors in \cite{AFH} provided a study of the resonances of $P$ using another method, which allowed to generalize the previous results to tangential crossings ($m \geq 2$ below). In short, this new method consists in reducing the problem to the proof of a stationary phase expansion with a degenerate phase, and it does not require the use of normal forms. In this article, we use this new method to get rid of the elliptic assumption on $U$, which was assumed up until now. The problem then reduces to the proof of a stationary phase expansion with a degenerate phase and a degenerate amplitude (see Appendix \ref{Appendix:DegenerateStationaryPhaseEstimates}). 
A similar study of the distribution of semiclassical resonances of $P$ has been conducted for an energy level at the same level as the crossing (referred to as cases \ref{Case:II} and \ref{Case:III} below) in \cite{FMW1} and \cite{FMW2} when potentials cross transversely ($m=2$) and in \cite{AFH2} when potentials cross tangentially ($m \geq 4$). We also generalize those results by removing the ellipticity assumption and reducing the proof to a degenerate stationary phase argument.

If the symbol $U(x,\xi) = r_0(x) + ihr_1(x) \xi$ of the interaction term vanishes identically on $\Rr$, then the study of the system $(P-E_0)u = 0$ for $P$ can be brought back to the two eigenvalue problems $(P_j-E_0)u = 0$ corresponding to each scalar operator $P_j$. However, the well known \textit{golden rule} of P. Dirac, coined by E. Fermi, states that the presence of the interaction $U$ between the two scalar potentials $P_j$ may create \textit{resonances} for the matrix operator $P$. More precisely, near an energy $E_0$, this interaction shifts the eigenvalues of $P_1$ away from the real axis, and the shifted points on the complex plane obtained are the resonances of $P$. It is well known that the width between the resonances and the real line is proportional to the inverse of the half-life of the molecule described by $P$; our goal is to compute this width. 

It has been established in \cite{AFH} that the width of the resonances gets larger as the crossings become more degenerate. This implies that the molecule has a shorter half-life.
Conversely, the main result of this article is that the width of the resonances gets narrower as the interaction term gets weaker, resulting in a longer half-life. More precisely we show that if $r_1=0$, $k_0$ is the vanishing order of $r_0$ and $m$ the contact order of the characteristic sets $\Gamma_j(E_0)$, then in the cases \ref{Case:Ia} and \ref{Case:Ib} the imaginary part of the resonances is of size $h^{1+2(k_0+1)/(m+1)}$ in the general case and of size $h^{1+2(k_0+2)/(m+1)}$ when both $k_0$ and $m$ are odd. In the case \ref{Case:II}, we show that the imaginary part of the resonances is of size $h^{1+2(2k_0+1)/(m+1)}$. We also compute the first terms of the associated asymptotic expansions. We also discuss the case where $U = r_0(x) + ir_1(x)hD_x$, which is relevant to the physical situation (see \cite{FMW2}). We show that the main results hold replacing $k_0$ by $\min(k_0,k_1)$ in cases \ref{Case:Ia} and \ref{Case:Ib}, and by $\min(k_0,k_1+1/2)$ in case \ref{Case:II}.

In Section \ref{Section:MainResult}, we introduce the framework of this article in the cases \ref{Case:Ia} and \ref{Case:Ib}, and we state in the main results the effect that the vanishing order of $r_0$ at $x=0$ has on the distribution of resonances. In Section \ref{Section:ConnectionFormula}, we state and prove the microlocal connection formula at a crossing point, which is the key result in the proof of the two main theorems. This essentially reduces to the degenerate stationary phase estimate and expansions proved in Appendix \ref{Appendix:DegenerateStationaryPhaseEstimates}; the use of those results in this context constitutes the main improvement of this article in regard to previous works. We then apply the microlocal connection formula in Section \ref{Section:ProofOfMainResult} to figure out the width of the resonances.
Next, we deal with the case \ref{Case:II} in Section \ref{Section:typeB}.
We finally discuss in Section \ref{Section:r1} the generalization (for all cases \ref{Case:Ia}, \ref{Case:Ib} and \ref{Case:II}) to the case where $U = r_0(x) + ihr_1(x)D_x$ and compare the influence of $r_0$ and $r_1$.

\section{Framework and main results}\label{Section:MainResult}
We make a few assumptions on the potentials $V_j$, the interaction term $U$ and the considered energy $E_0$. The situation is pictured in the figures \ref{Fig:PotentialsCrossings} and \ref{Fig:PhaseSpaceCrossings} 
. We define $P$ as in \eqref{Eq:MatrixSchrödingerOperator} with $U = r_0(x)$. 

\begin{assumption}\label{Assumption:AnalyticContinuation}
The potentials $V_1$, $V_2$ and the function $r_0$ are smooth and real-valued on $\Rr$. Moreover, there exist $\delta_0 > 0$ such that $V_1$, $V_2$ and $r_0$ have a bounded analytic continuation to the complex domain
\[
\Sigma := \{ x \in \Cc \; ; \; |\Im(x)| < \delta_0 \langle \Re(x) \rangle\}.
\]
where $\langle t \rangle := \sqrt{1+|t|^2}$. Moreover, the potentials $V_j$ admit limits $V_j^\pm := \lim_{\substack{\Re(x) \to \pm \infty\\x \in \Sigma}} V_j(x) \in \Rr$ different from $E_0$.
\end{assumption}

Under Assumption \ref{Assumption:AnalyticContinuation}, the operator $P$ is defined as an unbounded, essentially self-adjoint operator on $L^2(\Rr,\Cc^2)$ with domain $H^2(\Rr,\Cc^2)$. The analytic assumptions on $P$ allow to define the resonances of $P$ near $E_0$ using the analytic dilation method (see \cite{AC}). Consider the complex box 
\begin{equation}\label{Eq:ComplexBox}
\mc{R}_h := [E_0 - Lh, E_0 + Lh] + i [-Lh, Lh] \quad \text{with } L > 0 \text{ independent of } h.
\end{equation}
We define the resonances of $P$ near $E_0$ as follows.
\begin{definition}
We say that a solution $w$ of the equation $(P-E)w = 0$ is \textbf{outgoing} when, for some $\theta > 0$ sufficiently small, $w(xe^{i\theta}) \in L^2(\Rr,\Cc^2)$. Then, when $L$ is sufficiently small, we define $E \in \mc{R}_h$ to be a \textbf{resonance of $P$} when there exists an outgoing, non-identically vanishing solution $w$ of $(P-E)w = 0$. Such $w$ is then called a \textbf{resonant state}. We write $\Res_h(P)$ the set of resonances of $P$ in $\mc{R}_h$.
\end{definition}

\begin{rmk}\label{Rmk:OutgoingStates}
One convenient attribute of outgoing states is that they are microlocally infinitely small in the sense of \eqref{Eq:MicrolocallyInfinitelySmall} on the incoming tail ($\gamma_2^\flat$ with notations of \eqref{Eq:DefinitionIncomingPaths}). This fact is well known in the scalar case (see \cite{HS}), and the reduction from the matrix case to the scalar case is explained for example in \cite{FMW3}.
\end{rmk}

\begin{rmk}
The analycity in the angular complex domain 
\[
\{ x \in \Cc \; ; \;|\Im(x)| < \tan(\theta_0)|\Re(x)|\stext{and} |\Re(x)| > R\}
\]
with $\pi/2 > \theta_0 > 0$ and $R > 0$ is actually enough to define resonances via the analytic distortion method (see \cite{Hun}), which is more refined than the analytic dilation used above. We can also prove the microlocal connection formula (Theorem \ref{Thm:AsymptoticConnectionFormula}) as well as the main results of Section \ref{Section:MainResult} under this weaker assumption for the cases \ref{Case:Ia} and \ref{Case:Ib}. However in the case \ref{Case:II}, some technical difficulties appear when defining WKB solutions normalized near the turning point (see the phase functions $\phi_j^\bullet$ in \eqref{Eq:frhoTypeB}). These difficulties can be avoided by assuming the analycity around $0$ as in Assumption \ref{Assumption:AnalyticContinuation}.
\end{rmk}

\begin{assumption}\label{Assumption:TrappingAndNonTrapping}
There exist real numbers $a_0 \leq 0 \leq a'_0$ and $b_0$ such that $a_0 \neq a'_0$ and, for all $x \in \Rr$,
\begin{equation}\label{Eq:TrappingWellCondition}
\frac{V_1(x)-E_0}{(x-a_0)(x-a'_0)} > 0
\end{equation}
and
\begin{equation}\label{Eq:NonTrappingCondition}
\frac{V_2(x)-E_0}{x-b_0} > 0.
\end{equation}
\end{assumption}


\noindent
Remark that, under the two previous assumptions, we have $V_1^\pm > E_0$ and $V_2^+ > E_0 > V_2^-$.

We state some known results in spectral theory holding under Assumption \ref{Assumption:TrappingAndNonTrapping}. To this end, we first define $\Gamma_j = \Gamma_j(E) := p_j^{-1}(\{E\})$ the classical trajectory, or characteristic set, associated to the Hamiltonian $p_j(x,\xi) = \xi^2 + V_j(x)$ for the energy $E \in \Rr$. 

On one hand, under \eqref{Eq:NonTrappingCondition}, the spectrum of $P_2$ near $E_0$ is continuous. Moreover, the energies near $E_0$ are \textit{non-trapping}. This means that the Hamiltonian flow $e^{tH_p}$ carries any point $\rho \in \Gamma_2(E)$ of the characteristic set of $P_2$ away from any compact set: $\norm{e^{tH_p}\rho}{} \underset{t  \to \pm \infty}{\longrightarrow}+\infty$. Here, $H_p = \partial_\xi p \partial_x -\partial_x p \partial_\xi$ is the Hamiltonian vector field associated to $P$. It is known that there are no resonances for $P_2$ in a complex zone around $E_0$ of width $h \log\left( 1/h\right)$ (see \cite{Mar}). 

On the other hand, under \eqref{Eq:TrappingWellCondition}, the spectrum of $P_1$ consists of eigenvalues near $E_0$. 
Under the condition \eqref{Eq:TrappingWellCondition}, $V_1$ admits a simple well for energies $E \in \Rr$ close to $E_0$, that is to say $\Gamma_1(E)$ is a simple closed curve (or Jordan curve). 
To specify the eigenvalues of $P_1$, we first define the classical action on $\Gamma_1$ as
\begin{equation}\label{Eq:a_0(E)}
\mc{A}(E) := \int_{\Gamma_1(E)} \xi dx.
\end{equation}
which is a smooth and monotonically increasing\footnote{\; For $E \in \Rr$, the function $\mc{A}(E)$ describes the area enclosed by $\Gamma_1$ and $\mc{A}'(E) = \int_{\Gamma_1(E)} \frac{dx}{2\xi}$ corresponds to the time that any classical solution $e^{tH_{p_1}}(x_0,\xi_0)$ takes to travel $\Gamma_1$. This function analytically extends in $\Cc$ around $E = 0$.} function of $E$ near $E_0$. Then, it is well known (see \cite{ILR}, \cite{Yaf}) that the eigenvalues of $P_1$ can be approximated by energies satisfying the so-called Bohr-Sommerfeld quantization rule
\begin{equation}\label{Eq:Bohr-Sommerfeld quantization rule}
\cos\left( \frac{\mc{A}(E)}{2h} \right) = 0.
\end{equation}
We define accordingly the following subset of $[E_0-Lh, E_0+Lh]$, consisting of real energies close to $E_0$ which satisfy \eqref{Eq:Bohr-Sommerfeld quantization rule}:
\[
\mathfrak{B}_h:= \left\{E \in [ E_0 - Lh, E_0 + Lh]; \; \cos\left( \frac{\mc{A}(E)}{2h} \right) = 0\right\}.
\]

When $r_0$ vanishes identically on $\Rr$, the spectrum of $P$ near $E_0$ consists of a continuous spectrum with embedded eigenvalues, approximated by the Bohr-Sommerfeld quantization rule \eqref{Eq:Bohr-Sommerfeld quantization rule}. In the general case however, these embedded eigenvalues become resonances: this corresponds to Fermi's golden rule. 

\begin{assumption}\label{Assumption:FiniteOrderForPotentials}
The set $\{x \in \Rr; \; V_1(x) = V_2(x)\leq E_0\}$ of potential crossings is reduced to a unique point $\{0\}$ and the function $V_2-V_1$ vanishes at $x=0$ at a finite order $n \in \Nn \setminus \{0\}$:
\begin{equation}\label{Eq:ContactOrderOfPotentials}
(V_2-V_1)^{(j)}(0) = 0 \text{ for all }  j \in \{0, \dots, n-1\}, \stext{and} (V_2-V_1)^{(n)}(0) \neq 0.
\end{equation}
\end{assumption}

We list the four possible cases for potentials satisfying assumptions \ref{Assumption:TrappingAndNonTrapping} and \ref{Assumption:FiniteOrderForPotentials}:
\begin{enumerate}
    \item[\namedlabel{Case:Ia}{(I-a)}] $a_0 < 0 < b_0 < a'_0$ in which case $E_0 > 0$ and $n \geq 3$ is odd. 
    \item[\namedlabel{Case:Ib}{(I-b)}] $a_0 < 0 < a'_0 < b_0$ in which case $E_0 > 0$ and $n \geq 2$ is even. 
    \item[\namedlabel{Case:II}{(II)}] $a_0<0 = b_0 = a'_0$ in which case $E_0 = 0$. 
    \item[\namedlabel{Case:III}{(III)}] $a_0 = b_0 = 0 < a'_0$ in which case $n=1$. This is the setting of \cite{FMW1} and will be omitted here.  
\end{enumerate}


In the following, \textit{crossing points} refer to elements of $\ms{C} := \Gamma_1(E) \cap \Gamma_2(E)$, and in this setting there are two crossing points, noted $\rho_{\pm}$. The \textit{contact order} at a crossing point $\rho$ refers to the geometric contact order $m$ of the characteristic sets at $\rho$. It is defined as the order at which the functions $H_{p_1}^m p_2 (\rho)$ and $H_{p_2}^m p_1 (\rho)$ vanish at $\rho$. In this context, $m = n$. We will write $\Gamma_0(E) := \Gamma_1(E) \cup \Gamma_2(E)$ the union of the two classical trajectories. The \textit{turning points} will refer to the points of $\Gamma_0(E) \cap \{ \xi = 0\}$; they correspond to the points $a_0(E)$, $b_0(E)$ and $a'_0(E)$ defined in Assumption \ref{Assumption:TrappingAndNonTrapping}. Note that they differ from crossing points here. 

The figures \ref{Fig:PotentialsCrossings} and \ref{Fig:PhaseSpaceCrossings} account for the case \ref{Case:Ia}. Note that one can assume without loss of generality that $V_1(0) = V_2(0) = 0$  (replace $E_0$ by $E_0-V_1(0)$). 

\begin{figure}[h!]
    {\centering
    \begin{minipage}{0.45\textwidth}
    \centering
    \includegraphics[width=7cm]{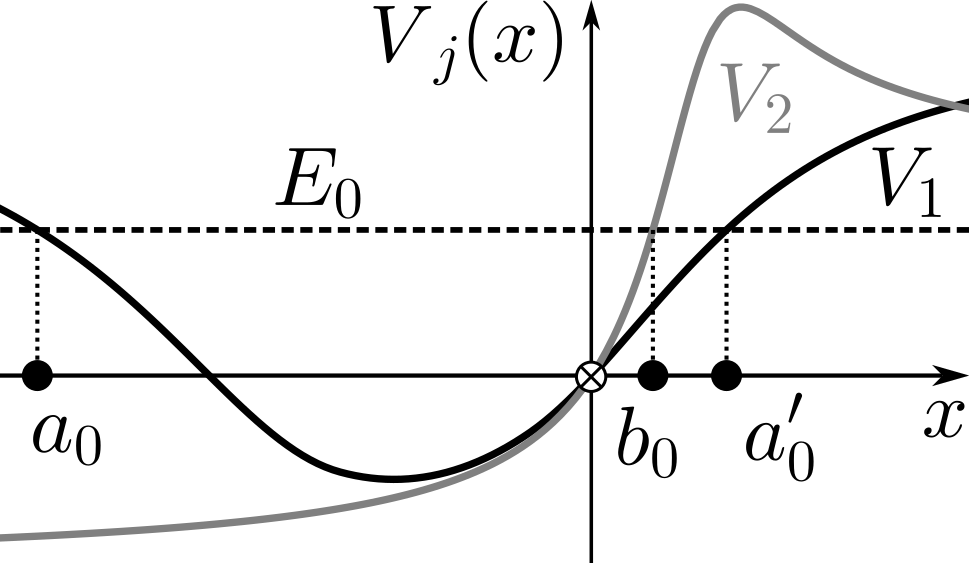}
    \caption{Potential crossing at $x=0$ - case \ref{Case:Ia}}
    \label{Fig:PotentialsCrossings}
    \end{minipage}\hfill
    \begin{minipage}{0.45\textwidth}
    \centering
    \includegraphics[width=7cm]{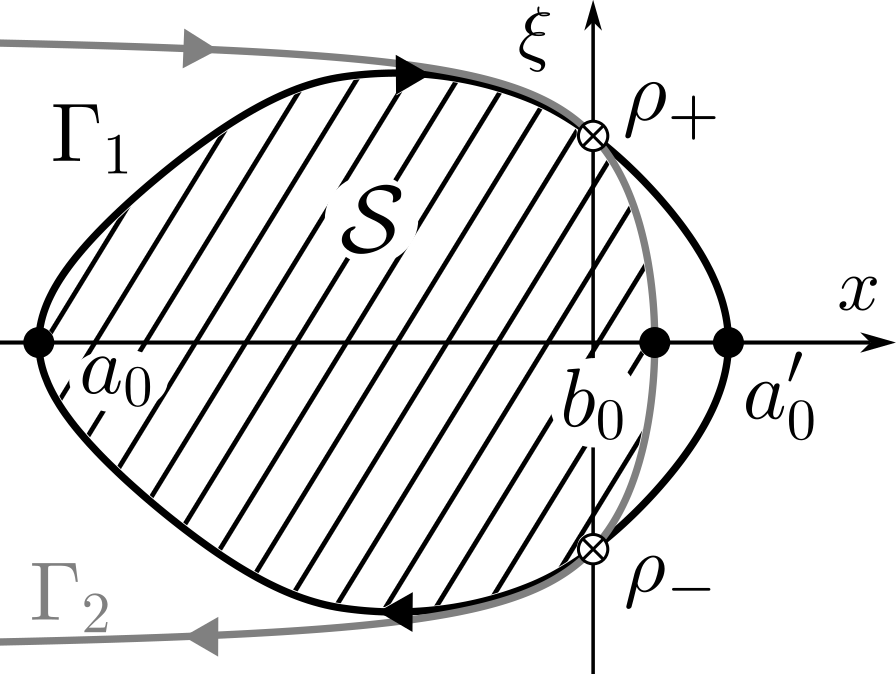}
    \caption{Phase space crossings of the associated classical trajectories case \ref{Case:Ia}}
    \label{Fig:PhaseSpaceCrossings}
    \end{minipage}
    }
\end{figure}

\begin{assumption}\label{Assumption:FiniteOrderInteraction}
The function $r_0$ vanishes at $x=0$ at a finite order $k \in \Nn$ : 
\begin{equation}\label{Eq:VanishingConditionOnInteractionTerm}
r_0^{(j)}(0) = 0 \stext{for all} j \in \{0, \dots, k-1\} \stext{and} r_0^{(k)}(0) \neq 0
\end{equation}
Moreover, we assume that $k < n$.
\end{assumption}

The assumption $k < n$ ensures that the principal terms appearing  in the asymptotic expansions (terms of order $h^{(k+1)/(m+1)}$ in Section \ref{Section:ConnectionFormula} where $m=n$ or of order $h^{(2k+1)/(m+1)}$ in Section \ref{Section:typeB} where $n = m/2$) do not get absorbed by the approximations of order $h$ that we use throughout the article (ex: asymptotic behavior of scalar solutions \eqref{Eq:AsymptoticBehaviorScalarSolutions} used in the proof of Proposition \ref{Prop:IntegralOperators}).

The main results are the two following theorems; they will be proved in Section \ref{Section:ConnectionFormula}. The first one gives the asymptotic distribution of resonances of the matrix operator $P$ in terms of geometric data at the crossing points and in terms of the actions of the closed trajectories. In particular, the width of resonances gets narrower as $k$ grows and larger as $m$ grows.

\begin{thm}\label{Thm:MainResult}
Assume that assumptions \ref{Assumption:AnalyticContinuation} up to \ref{Assumption:FiniteOrderInteraction} hold true and that either \ref{Case:Ia} or \ref{Case:Ib} holds true. Then for all small $h >0$, there exists a bijective map $z_h: \mathfrak{B}_h \to \Res_h(P)$ such that for any $E \in \mathfrak{B}_h$ one has 
\begin{equation}\label{Eq:zh(E)}
\va{z_h(E) - E} = \BigO{h^{1+2\frac{k+1}{m+1}}}
\end{equation}
and
\begin{equation}\label{Eq:Im(E)}
\Im z_h(E) = -D(E)h^{1+2\frac{k+1}{m+1}} + \BigO{h^{1+2\frac{k+1}{m+1} + s}}
\end{equation}
with $s := \min\left( 1/3, 1/(m+1)\right)$ and
\begin{equation}\label{Eq:D(E)}
D(E) := \frac{2|\omega|^2}{\mc{A}'(E_0)}  \va{ \sin\left(\frac{\sgn(v_m)(k+1)}{2(m+1)}\pi + \frac{\mc{S}(E)}{2h}\right)}^2.
\end{equation}
Here 
$\mc{S} = \mc{S}(E) := 2\left(\int_{a'_0(E)}^0 \sqrt{E-V_1(x)}dx + \int_0^{b_0(E)} \sqrt{E-V_2(x)}dx\right)$
is the area bounded by $\Gamma_1(E)$ and $\Gamma_2(E)$ (shaded area on Figure \ref{Fig:PhaseSpaceCrossings}), and $\omega$ is defined by
\begin{equation}\label{Eq:Omega}
\omega := \frac{\mu_{k,m}\left(\frac{\sgn(v_m)(k+1)}{2(m+1)}\pi\right)}{(m+1)k!}  
\pmb\Gamma\left(\frac{k+1}{m+1}\right)  E_0^{\frac{k-m}{2(m+1)}} \left(\frac{2(m+1)!}{\va{v_m}}\right)^{\frac{k+1}{m+1}}r_0^{(k)}(0)
\end{equation}
where $v_m := (V_2-V_1)^{(m)}(0)$ and
\begin{equation}\label{Eq:mu_kmtheta}
\mu_{k,m}(\theta) := \frac{e^{i\theta} + (-1)^ke^{i(-1)^{m+1}\theta}}{2} \stext{for $\theta$ in $\Rr$.}
\end{equation}
\end{thm}

In the case where both $k$ and $m$ are odd, $D(E)$ vanishes and the imaginary part of the resonances shrinks, as stated in the following theorem.

\begin{thm}\label{Thm:MainResultOdd}
Assume that assumptions \ref{Assumption:AnalyticContinuation} up to \ref{Assumption:FiniteOrderInteraction} and \ref{Case:Ib} hold true. Also assume that $k+1 < m$. Then, for all small $h >0$, there exists a bijective map $z_h: \mathfrak{B}_h \to \Res_h(P)$ such that for any $E \in \mathfrak{B}_h$ one has 
\begin{equation}\label{Eq:zh(E)Odd}
\va{z_h(E) - E} = \BigO{h^{1+2\frac{k+2}{m+1}}}
\end{equation}
and
\begin{equation}\label{Eq:Im(E)Odd}
\Im z_h(E) = -D_{odd}(E)h^{1+2\frac{k+2}{m+1}} + \BigO{h^{1+2\frac{k+2}{m+1} + s}}
\end{equation}
with $s$ as in Theorem \ref{Thm:MainResult} and
\begin{equation}\label{Eq:D(E)Odd}
D_{odd}(E) := \frac{2|\omega_{odd}|^2}{\mc{A}'(E_0)}  \va{ \sin\left(\frac{\sgn(v_m)(k+2)}{2(m+1)}\pi + \frac{\mc{S}(E)}{2h}\right)}^2.
\end{equation}
Here, $\omega_{odd}$ is defined by
\begin{equation}\label{Eq:OmegaOdd}
\omega_{odd} := 
\frac{\mu_{k+1,m}\left(\frac{\sgn(v_m)(k+2)}{2(m+1)}\pi\right)}{(m+1)(k+1)!} 
\pmb\Gamma\left(\frac{k+2}{m+1}\right)
E_0^{\frac{k+1-m}{2(m+1)}}
\left(\frac{2(m+1)!}{\va{v_m}}\right)^{\frac{k+2}{m+1}}
\left(\tau r_0^{(k)}(0) + r_0^{(k+1)}(0)\right)
\end{equation}
where
\begin{equation}\label{Eq:nu}
\tau := \frac{V_1'(0)}{2} \left(k+1-\frac{2(2k+1)\sgn(v_m)}{(m+2)} \right)
- \frac{(2k+1)v_{m+1}}{(m+1)(m+2)\va{v_m}}.
\end{equation}
\end{thm}

The leading coefficient $D(E)$ (resp. $D_{odd}(E)$) of $\Im z_h(E)$ may vanish at some energies $E\in \mathfrak{B}_h$, for each (small) $h$, more specifically energies $E$ such that $2\arg(\omega) - h^{-1}\mc{S}(E) \in \pi + 2\pi \Zz$. In this case, we observe resonances that are closer to the real axis than those with imaginary part of order $h^{1+2(k+1)/(m+1)}$ (resp. $h^{1+2(k+2)/(m+1)}$). This phenomenon, called \textit{double line of resonances} in \cite{AFH}, occurs due to the existence of a \textit{directed cycle} in $\Gamma(E_0)$. This directed cycle is the (generalized) trajectory born from the combined classical actions of both $P_1$ and $P_2$. On Figure \ref{Fig:PhaseSpaceCrossings}, it corresponds to the simple closed curve linking $\rho_+$ and $\rho_-$ via $a_0$ and $b_0$. It is illustrated in \cite[Example 2.3]{AFH}.

\begin{rmk}
As pointed in \cite[Remark 2.1]{AFH2}, the limit $h \to 0^+$ should be taken in a subset $\mathfrak{J} \subset(0,1]$, $0 \in \overline{\mathfrak{J}}$, otherwise the bijectivity of $z_h$ may break down.
\end{rmk}

\section{Microlocal connection formula at a crossing point}\label{Section:ConnectionFormula}
\subsection{Definition of the transfer matrix}\label{SubSection:DefinitionTransferMatrix}

We briefly recall some definitions of semiclassical analysis (see \cite{Zwo}). 
In this paper, we say that a function $f \in L^2(\Rr,\Cc^2)$ depending on $h$ is microlocally infinitely small around a point $(x_0,\xi_0) \in \Rr^2$, and note
\begin{equation}\label{Eq:MicrolocallyInfinitelySmall}
	f \equiv 0 \stext{near} (x_0,\xi_0)
\end{equation}
if there exists $\chi \in \C^\infty_b(\Rr^2, \Rr)$ independent of $h$ such that $\chi(x_0,\xi_0) \neq 0$ and
\[
\norm{\chi^W(x,hD_x)f}{L^2(\Rr,\Cc^2)} = \BigO{h^\infty}.
\]
We say that $f \equiv 0$ near $\Omega$, with $\Omega \subset \Rr^2$, when the previous definition applies to every point of $\Omega$. 

In this section, we fix $E_0 \in \Rr$ satisfying the previous assumptions and $E \in \mc{R}_h$. Our goal for this section is to define the transfer matrix at each crossing point (microlocal connection formula) for such energies $E$. This requires to recall a few results on microlocal solutions.

\noindent
We briefly introduce a construction (based on a WKB method) of microlocal solutions on the characteristic curves $\Gamma_j$ away from crossing and turning points. 
First, we recall that the space of microlocal solutions to the equation $(P-E)w \equiv 0$ near any connected subset of $\Gamma \setminus \ms{C}$ is of dimension $1$.
This is due to the fact that a scalar symbol of real principal type at a given point (a symbol $p$ such that $p = 0$ and $H_p \neq 0$ at this point) has its characteristic set of dimension $1$ locally around that point. In our case, $p_j - E$ is of principal type at all points of $\Gamma_j(E) \setminus \ms{C}$. The matrix case can be brought to this scalar case by transforming $P$ via a canonical (symplectic) transformation. 
Next, fix $j \in \{1,2\}$ and any point $\rho_j = (x_j,\xi_j) \in \Gamma_j(E_0) \setminus (\ms{C} \cup \{ \xi = 0\})$ on $\Gamma_j$ which is neither a crossing point nor a turning point. Then, we construct via a WKB method a non-trivial microlocal solution $f_{\rho_j}$ to the equation $(P-E)f_{\rho_j} \equiv 0$ near $\rho_j$ (see for example \cite[Proposition 5.4]{FMW3}) of the form
\begin{equation}\label{Eq:frho}
	f_{\rho_j}(x,h) = e^{\frac{i}{h}\sgn(\xi_j)\phi_j(x)}\begin{pmatrix}
		\sigma_{j,1}(x,h)\\
		\sigma_{j,2}(x,h)
	\end{pmatrix},
\end{equation}
where the phase function $\phi_j$ is defined starting from the $x$-coordinate $\rho_x = 0$ of the crossing points $(0,\pm \sqrt{E_0})$ by $\phi_j(x) := \int_0^x\sqrt{E_0-V_j(y)}\,dy$ and $\sigma_{j,k}(x,h)\sim\sum_{l\ge0}h^l\sigma_{j,k,l}(x)$ with
\begin{equation*}
	\sigma_{j,j,0}(x)=(E_0-V_j(x))^{-1/4}, \quad\sigma_{j,\hat j,0}(x)=0 \stext{and} \sigma_{j,j,l}(x_j)=0\quad\stext{for all} l\ge1.
\end{equation*}
Here, $\hat{j}$ is the complementary of $j$ in $\{1,2\}$, meaning that $\{j,\hat{j}\}=\{1,2\}$. 

Secondly, we study the behavior of microlocal solutions near crossing points. We split $\Gamma \setminus (\ms{C} \cup \{ \xi = 0\})$ in eight connected components according to the crossing points and the turning points (see Figure \ref{Fig:CrossingsPaths} for the labelling at $\rho_+$):
\begin{align}
	&\gamma_j^\flat:=\Gamma_j\cap\{x<0, \xi>0\}, \quad
	\gamma_j^\sharp:=\Gamma_j\cap\{x>0, \xi>0\},
	\label{Eq:DefinitionOutgoingPaths}\\
	&\breve\gamma_j^\flat:=\Gamma_j\cap\{x>0, \xi<0\}, \quad 
	\breve\gamma_j^\sharp:=\Gamma_j\cap\{x<0, \xi<0\}.\label{Eq:DefinitionIncomingPaths}
\end{align}
\begin{figure}[h!]
	\centering
	{
		\includegraphics[width=5.1cm]{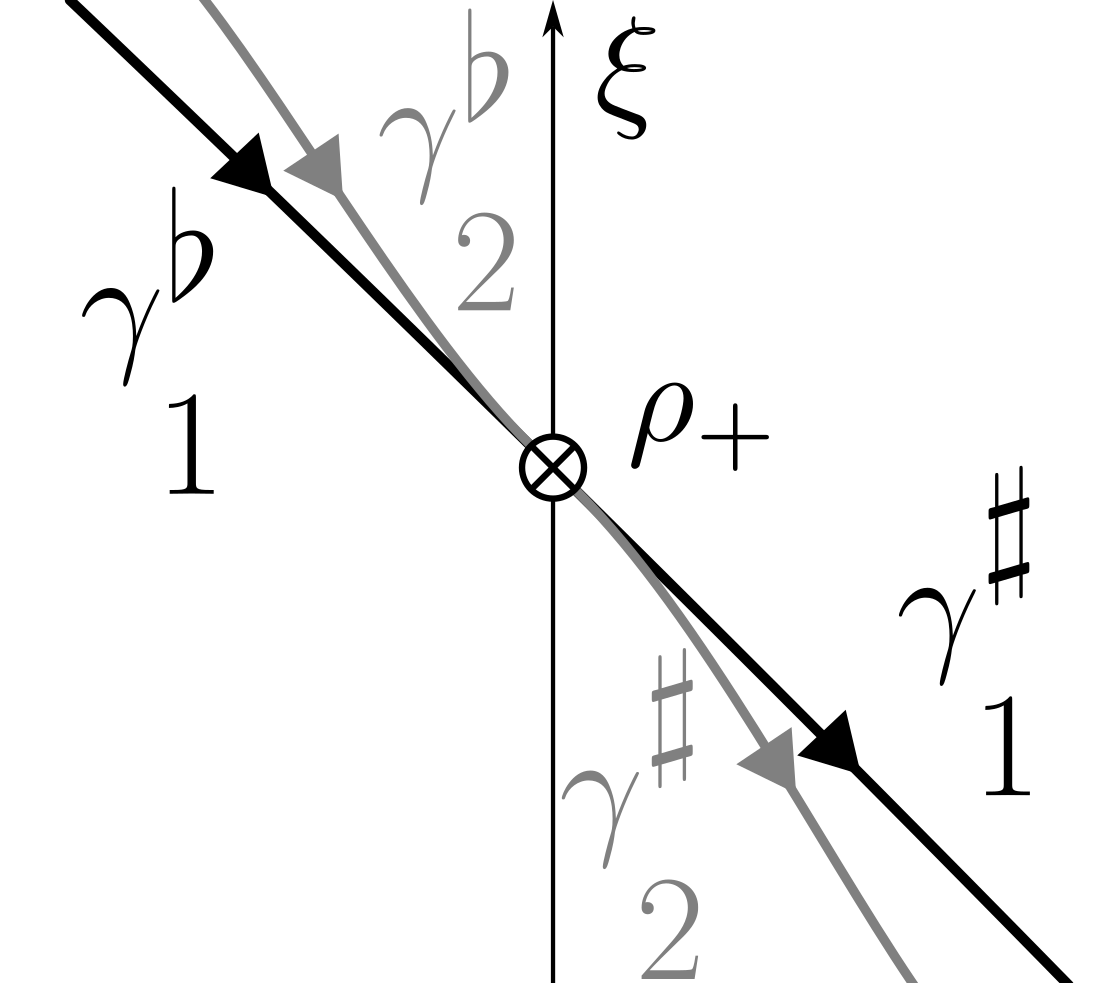}
		\caption{Labeling around the upper crossing point $\rho_+$}
		\label{Fig:CrossingsPaths}
	}
\end{figure}

Let $\rho_j^\bullet$ be any point of $\gamma_j^\bullet$.  
Since the space of microlocal solutions away from crossing points is of dimension $1$, any microlocal solution $w \in L^2(\Rr,\Cc^2)$ of
\[
(P-E)w \equiv 0 \stext{near} \rho_\pm
\]
is a (complex) multiple
\[
w \equiv \alpha_j^\bullet f_j^\bullet \stext{near} \rho_j^\bullet
\]
of the previously introduced microlocal solution $f_j^\bullet:= f_{\rho_j^\bullet}$. We symmetrically define $\breve\alpha_j^\bullet$ and $\breve f_j^\bullet$ on $\breve\gamma_j^\bullet$. The following theorem sums this remark up and gives an asymptotic microlocal connection formula near the crossing points.
\begin{thm}\label{Thm:AsymptoticConnectionFormula}
	Assume that assumptions \ref{Assumption:AnalyticContinuation} up to \ref{Assumption:FiniteOrderInteraction} hold true and that either \ref{Case:Ia} or \ref{Case:Ib} holds true. The space of microlocal solutions $w$ of the equation $(P-E)w \equiv 0$ near the crossing point $\rho_\pm$ is of dimension $2$ : there exists a $2 \times 2$ matrix $T^\pm = (t_{i,j}^\pm)_{1 \leq i,j \leq 2}$ depending on $E$, $h$ and $\rho_j^\bullet$ (respectively $E$, $h$ and $\breve\rho_j^\bullet$ for $T^-$) such that, for any microlocal solution $w \in L^2(\Rr,\Cc^2)$ of $(P-E)w \equiv 0$ near $\rho_\pm$ with $\norm{w}{L^2(\Rr,\Cc^2)} \leq 1$,
	\begin{equation}\label{Eq:TransferMatrix}
		\begin{pmatrix}
			\alpha_1^\sharp\\
			\alpha_2^\sharp
		\end{pmatrix}
		= T^+(E,h)
		\begin{pmatrix}
			\alpha_1^\flat\\
			\alpha_2^\flat
		\end{pmatrix}
		\stext{and}
		\begin{pmatrix}
			\breve\alpha_1^\sharp\\
			\breve\alpha_2^\sharp
		\end{pmatrix}
		= T^-(E,h)
		\begin{pmatrix}
			\breve\alpha_1^\flat\\
			\breve\alpha_2^\flat
		\end{pmatrix}
	\end{equation}
	We call \textbf{transfer matrix at} $\pmb\rho_\pm$ the matrix $T^\pm$. Moreover, $T^+$ is given for all $E \in \mc{R}_h$ by
	\begin{equation}\label{Eq:AsymptoticExpansionTransferMatrix}
		T^+(E,h) = I_2 -i h^{\frac{k+1}{m+1}}\begin{pmatrix}
			0&\overline{\omega}\\
			\omega&0
		\end{pmatrix}
		+ \BigO{h^{\frac{k+2}{m+1}}}
	\end{equation}
	where $\omega$ is given by \eqref{Eq:Omega}. In the case where both $k$ and $m$ are odd and $k+1 < m$ we have
	\begin{equation}\label{Eq:AsymptoticExpansionTransferMatrixOdd}
		T^+(E,h) = I_2 -i h^{\frac{k+2}{m+1}}\begin{pmatrix}
			0&\overline{\omega_{odd}}\\
			\omega_{odd}&0
		\end{pmatrix}
		+ \BigO{h^{\frac{k+3}{m+1}}}
	\end{equation}
	where $\omega_{odd}$ is given by \eqref{Eq:OmegaOdd}. Finally, $\omega$ and $\omega_{odd}$ do not depend on the point $\rho_j^\bullet$ chosen. 
\end{thm}

\noindent
The asymptotic expansion in Theorem \ref{Thm:AsymptoticConnectionFormula} will be proved in the next subsection. 
A symmetric expansion for the transfer matrix $T^-$ at $\rho_-$ is obtained from \eqref{Eq:SymmetryTransferMatrices}. Concerning the first part of the theorem (dimension of the space of microlocal solutions), we refer to \cite[Section 3.4]{AFH}.

\subsection{Microlocal connection formula at a crossing point}\label{Subsection:ConnectionFormula}

Let us now turn to the proof of Theorem \ref{Thm:AsymptoticConnectionFormula}. We study the behavior of microlocal solutions on a small, compact, $h$-independent interval $I$ around the $x$-coordinate $x = 0$ of the crossing points. In the phase space, this means that we consider $\Gamma_j \cap (I \times \Rr_\xi)$. 
Unfortunately, we cannot directly use a WKB method as in \eqref{Eq:frho} to construct microlocal solutions near crossing points. We follow the idea of \cite{FMW1} and construct exact solutions for the scalar equations of $(P_j-E)u_j = 0$ based on Picard's successive approximation method (see for example \cite[Chapter 5]{Olv}). 
We obtain pairs ($u_j$, $\breve u_j$) of solutions to $(P_j-E)u_j = 0$ in I 
which asymptotically behave as
\begin{equation}\label{Eq:AsymptoticBehaviorScalarSolutions}
	u_j(x) = \frac{1+\BigO{h}}{\sqrt[4]{E-V_j(x)}}e^{\frac{i}{h}\int_0^y \sqrt{E - V_j(y)}dy}, \quad 
	\breve u_j(x) = \frac{1+\BigO{h}}{\sqrt[4]{E-V_j(x)}}e^{-\frac{i}{h}\int_0^y \sqrt{E - V_j(y)}dy}
\end{equation}
where $j \in \{1,2\}$. This asymptotic behavior when $h\to 0^+$ is uniform on $I$. Next, we use those scalar solutions to construct vector-valued solutions $w$ of the system $(P-E)w = 0$. 
Those solutions, constructed with the appropriate 
integral operators, are close at order $\BigO{h}$ to the solutions $f_j^\bullet$ and small at order $\BigO{h}$ elsewhere (see \cite[Section 4]{FMW1}). Essentially, this construction relies on the estimates of Proposition \ref{Prop:IntegralOperators} below.

We consider the Banach space $\C(I) := \C(I,\Cc)$ of continuous functions on $I$ endowed with the $L^\infty$ norm. We fix a small $\delta >0$ and define for $f \in \C(I)$ and $x \in I$ the integral operators $\widetilde{K}_{j,S}$ by
\begin{equation}\label{Eq:DefinitionIntegralOperators}
	\widetilde{K}_{j,S}f(x) := -\frac{1}{h\mc{W}(u_j,\breve u_j)}\left(\breve u_j(x)\int_{I_S(x)} u_j(y)r_0(y)f(y)dy - u_j(x)\int_{I_S(x)} \breve u_j(y)r_0(y)f(y)dy\right)
\end{equation}
where $I_R(x) := [\delta,x]$ and $I_L(x) := [-\delta,x]$ with $j \in \{1,2\}$ and $S \in \{L,R\}$. Here $\mc{W}$ stands for the wronskian and one can check that $\mc{W}(u_j,\breve u_j) = 2ih^{-1}$. The next proposition is a refinement of \cite[Proposition 3.10]{AFH}.

\begin{prop}\label{Prop:IntegralOperators}
	For all $j \in \{1,2\}$ and $S \in \{L,R\}$, $\widetilde{K}_{j,S}$ is a bounded operator of $\C(I)$. Moreover, 
	\begin{equation}\label{Eq:EstimateOnKernelsKtilde}
		\norm{\widetilde{K}_{j,S}\widetilde{K}_{{\hat j},S}}{\mc{B}(\C(I))} = \BigO{h^{\frac{k+1}{m+1}}}.
	\end{equation}
\end{prop}

\begin{proof}
	First, $\widetilde{K}_{j,S}$ is bounded since $r_0$, $u_j$ and $\breve u_j$ are bounded on $I$. We reduce the proof to the study of 
	\[
	\int_{-\delta}^x u_1(y) r_0(y) \breve u_2(y) \left(\int_{-\delta}^y u_2(z) r_0(z) f(z) dz\right) dy
	\]
	for $f \in \C(I)$, because the other integrals appearing when expanding $\widetilde{K}_{j,S}\widetilde{K}_{{\hat j},S}f$ can be treated in a symmetric way. We first note $a_j$ (resp. $\breve a_j$) the amplitude of $u_j$ (resp. $\breve u_j$) defined in \eqref{Eq:AsymptoticBehaviorScalarSolutions}. Both $a_j$ and $\breve a_j$ behave as $\frac{1+\BigO{h}}{\sqrt[4]{E-V_j(y)}}$ when $h \to 0^+$.
	We then apply \eqref{Eq:DegenerateStationaryPhaseEstimate} of Lemma \ref{Lemma:DegenerateStationaryPhase} with
	$a(y) := {a}_1(y)\breve{a}_2(y)r_0(y)F(y), \quad$ $\phi(y) := \int_0^y \sqrt{E-V_1(z)} - \sqrt{E-V_2(z)}dz$
	and $F(y):= \int_{-\delta}^y u_2(z) r_0(z) f(z) dz$. 
\end{proof}

\begin{rmk}
	In practice, we work with $h$-dependent amplitudes (see \eqref{Eq:AsymptoticBehaviorScalarSolutions}). However they are always of the form $a(y,h) = a_0(y) + ha_1(y)$, where $a_0$ and $a_1$ are smooth functions, so that the Lemma \ref{Lemma:DegenerateStationaryPhase} can still be applied.
\end{rmk}

Now, following the method of \cite[Section 4]{FMW1}, we construct two basis $(\breve w_1^\sharp,w_1^\flat,\breve w_2^\sharp,w_2^\flat)$ and $(\breve w_1^\flat,w_1^\sharp,\breve w_2^\flat,w_2^\sharp)$ of exact solutions on $I$ to $(P-E)w = 0$  and a matrix $A = (a_{i,j}(E,h))_{1 \leq i,j \leq 4}$ such that
\[
(\breve w_1^\sharp,w_1^\flat,\breve w_2^\sharp,w_2^\flat) = A(\breve w_1^\flat,w_1^\sharp,\breve w_2^\flat,w_2^\sharp)
\]
and
\[
T^+ = \begin{pmatrix}
	a_{22}&a_{24}\\
	a_{42}&a_{44}\\
\end{pmatrix} + \BigO{h}.
\]
\noindent
More precisely, $A$ is of the form $A = (I_4-A_R)^{-1}(I_4-A_L)$, where the coefficients of $A_S$, $S \in \{R,L\}$, are of one of the two following types:
\begin{align}
	\alpha_{j,S}(f) := \pm\frac{i}{2} \int_{I_S(0)} \breve u_{\hat j} U_{\hat{j}} J_{j,S} f \stext{ or } \beta_{j,S}(f) := \pm\frac{i}{2} \int_{I_S(0)} \breve u_{\hat j} U_{j} \widetilde{K}_{\hat j,S} J_{j,S} f
\end{align}
where $J_{j,S}:= \sum_{l \geq 0} \left(\widetilde{K}_{j,S}\widetilde{K}_{\hat{j},S}\right)^l$ (when defined), $j \in \{1,2\}$ and $I_S$ is defined in \eqref{Eq:DefinitionIntegralOperators}. The $\pm$ sign is $+$ for $S=R$ and $-$ for $S=L$. We define symmetrically $\breve\alpha_{j,S}(f)$ and $\breve\beta_{j,S}(f)$ replacing $\breve u_j$ by $u_j$. The proof of Theorem \ref{Thm:AsymptoticConnectionFormula} is then reduced to a stationary phase argument as summed up in the next lemma. 

\begin{lemma}\label{Lemma:AlphaBetaEstimates}
	The matrix $A$ is of the form $A = \begin{pmatrix}
		I_2 & \widetilde A_1\\
		\widetilde A_2 & I_2
	\end{pmatrix}
	+ \BigO{h^{\frac{k+2}{m+1}}}$
	where $\widetilde A_j = \begin{pmatrix}
		\breve F_j\breve u_j&\breve F_ju_j\\
		F_j\breve u_j& F_ju_j\\
	\end{pmatrix}$ and $F_j:=\alpha_{j,L}-\alpha_{j,R}$, $\breve F_j:=\breve\alpha_{j,L}-\breve\alpha_{j,R}$. In particular the coefficients of $T^+$ are given by $a_{22}=a_{44}=1$ and
	\begin{equation*}
		a_{42} = F_1u_1 = -i\omega h^{\frac{k+1}{m+1}}+\BigO{h^{\frac{k+2}{m+1}}},\quad
		a_{24} = F_2u_2 = -i\overline{\omega} h^{\frac{k+1}{m+1}}+\BigO{h^{\frac{k+2}{m+1}}},
	\end{equation*}
	where $\omega$ is given by \eqref{Eq:Omega}. In the case where $k$ and $m$ are odd and $k+m<1$, $A = \begin{pmatrix}
		I_2 & \widetilde A_1\\
		\widetilde A_2 & I_2
	\end{pmatrix}
	+ \BigO{h^{\frac{k+3}{m+1}}}$ and 
	\begin{equation*}
		a_{42} = F_1u_1 = -i\omega_{odd} h^{\frac{k+2}{m+1}}+\BigO{h^{\frac{k+3}{m+1}}},\quad
		a_{24} = F_2u_2 = -i\overline{\omega_{odd}} h^{\frac{k+2}{m+1}}+\BigO{h^{\frac{k+3}{m+1}}},
	\end{equation*}
	where $\omega_{odd}$ is given by \eqref{Eq:OmegaOdd}.
\end{lemma}

\begin{proof}
	We look at $F_1u_1$ ; the other terms can be treated in a symmetric way. We have	
	\begin{equation}\label{Eq:F1u1}
		F_1u_1 = \frac{-i}{2}
		\left[\int_{-\delta}^\delta \breve u_2r_0u_1\right.\notag 
		+ \sum_{l \geq 1} \left(
		\int_{-\delta}^0 \breve u_2 r_0 (\widetilde{K}_{1,L} \widetilde{K}_{2,L})^lu_1 
		+ \int_{0}^\delta \breve u_2 r_0 (\widetilde{K}_{1,R} \widetilde{K}_{2,R})^lu_1 
		\right)
	\end{equation}
	where, as in the previous proposition, $a_j$ and $\breve a_j$ are the amplitudes of $u_j$ and $\breve u_j$ defined in \eqref{Eq:AsymptoticBehaviorScalarSolutions}.
	We start with the case where either $k$ or $m$ is even. We split \eqref{Eq:F1u1} into three parts:  first, the terms for $l \geq 2$ can be brutally bounded. Indeed, in view of Proposition \ref{Prop:IntegralOperators}, the geometric series $J_{j,S}$ is convergent for $h > 0$ small enough and $\sum_{l \geq 2} \norm{\left(\widetilde{K}_j\widetilde{K}_{\hat j}\right)^l}{\mc{B}(\C(I))} = \BigO{h^{2(k+1)/(m+1)}}$ where $\mc{B}(\C(I))$ is the Banach space of the bounded operators on $\C(I)$. 
	Next, the term for $l=1$ can be treated in a similar way as in Proposition \ref{Prop:IntegralOperators}. 
	A first application of \eqref{Eq:DegenerateStationaryPhaseEstimate} of Lemma \ref{Lemma:DegenerateStationaryPhase} while keeping only the term of order $h^{(k+1)/(m+1)}$ gets us smooth functions $C_S(y,h)$ bounded in $h$ such that $(\widetilde{K}_{1,S} \widetilde{K}_{2,S})u_1 (y) =  C_S(y,h)h^{(k+1)/(m+1)}$. We then apply a second time the same result and obtain the claimed estimate. 
	In a similar way, when $k$ and $m$ are odd, the term corresponding to $l=2$ is treated using \eqref{Eq:DegenerateStationaryPhaseEstimate} of Lemma \ref{Lemma:DegenerateStationaryPhase}, and the terms corresponding to $l \geq 3$ can be dealt with using Proposition \ref{Prop:IntegralOperators}.
	
	Finally, $\omega$ (respectively $\omega_{odd}$) is the first term of the asymptotic expansion of $\frac{1}{2}\int_{-\delta}^{\delta} \breve u_2 r_0 u_1$, which corresponds to the term $l=0$. We use \eqref{Eq:DegenerateStationaryPhase3Cases} of Lemma \ref{Lemma:DegenerateStationaryPhase} with $\phi(y) := \int_0^y \sqrt{E-V_1(z)} - \sqrt{E-V_2(z)}dz$ and $a := a_1\breve a_2$. We write a Taylor expansion at $y=0$ for both functions; we get
	\[
		\phi(y) = \frac{1}{2\sqrt{E}} \left[ v_m\frac{y^{m+1}}{(m+1)!} + \left( v_{m+1} + \frac{(m+1)v_m(V_1+V_2)'(0)}{4}\frac{y^{m+2}}{(m+2)!}\right)\right] 
		 + \BigO{y^{m+3}}
	\]
	and
	\[
		a(y) = \frac{1}{\sqrt{E}} \left[ r_0^{(k)}(0)\frac{y^{k}}{k!} + \left((k+1)\frac{(V_1+V_2)'(0)}{4}r_0^{(k)}(0) + r_0^{(k+1)}(0)\right)\frac{y^{k+1}}{(k+1)!}\right]
		+ \BigO{y^{k+2}}.
	\]
	To conclude, we use the fact $E = E_0 + \BigO{h}$ and the assumption $k < m$ so that $h = \Smallo{h^{(k+1)/(m+1)}}$ is negligible compared to the principal term $-i\omega h^{(k+1)/(m+1)}$ (in the odd case, we use $k < m+1$ for the same reason). 
\end{proof}
\begin{rmk}
	At this point, we can compute 
	$T^- = \begin{pmatrix}
		a_{11}&a_{13}\\
		a_{31}&a_{33}\\
	\end{pmatrix} + \BigO{h}$
	with symmetrical computations as previously. However, knowing $T^+$ we can also directly deduce $T^-$ by remarking that 
	\begin{equation}\label{Eq:SymmetryTransferMatrices}
		T^-(E,h) = \left(\overline{T^+(\overline{E},h)}\right)^{-1}.
	\end{equation}
\end{rmk}


\section{Proof of the main results}\label{Section:ProofOfMainResult}
We first recall \cite[Proposition 7.1]{FMW3}). If $E$ is a resonance of $P$ and $u$ an associated non-trivial resonant state, then for all $x_0 < a_0$ and with the notations of Theorem \ref{Thm:AsymptoticConnectionFormula},
\begin{equation}\label{Eq:Im(E)ExpressedWithuandalpha}
\Im(E) = -\frac{|\breve \alpha_2^\sharp|^2}{\norm{u}{L^2(x_0, +\infty)}^2}h(1+\BigO{h}).
\end{equation}

The proof is then reduced to computing an asymptotic expansion of $|\breve \alpha_2^\sharp|^2$ and $\norm{u}{L^2(x_0, +\infty)}$.

\begin{proof}[Proof of Theorem \ref{Thm:MainResult}]
For the time being, we assume that \eqref{Eq:zh(E)} holds true. Let $E \in \Res_h(P)$ be a resonance of $P$ and $u$ a non-trivial outgoing resonant state. We take back the notations from the microlocal connection formula (Theorem \ref{Thm:AsymptoticConnectionFormula}) and write
\begin{equation}\label{Eq:NormalizationForResonantState}
u \equiv \alpha_j^\bullet f_j^\bullet \stext{near} \rho_+ \stext{and} u \equiv \breve\alpha_j^\bullet  \breve f_j^\bullet \stext{near} \rho_-.
\end{equation}
Since $u$ is outgoing, we have $\alpha_2^\flat = 0$ (see Remark \ref{Rmk:OutgoingStates}). Furthermore, $u$ is non-trivial so we do not have $u \equiv 0$ near $\Gamma$. In particular $\alpha_1^\flat \neq 0$ (using the transfer matrix and Maslov correction on $\Gamma$, $\alpha_1^\flat = 0$ would imply $u \equiv 0$ near $\Gamma$). We can then assume $\alpha_1^\flat = 1$ even if it means replacing $u$ by $(\alpha_1^\flat)^{-1}u$. We now  compute $\breve \alpha_2^\sharp$. Using the transfer matrices $T^\pm$; we get
\[
\begin{pmatrix}
\alpha_1^\sharp\\
\alpha_2^\sharp
\end{pmatrix}
=
T^+ \begin{pmatrix}
\alpha_1^\flat\\
\alpha_2^\flat
\end{pmatrix}
=
\begin{pmatrix}
	t_{11}^+\\
	t_{21}^+
\end{pmatrix}
\stext{and}
\begin{pmatrix}
	\breve\alpha_1^\sharp\\
	\breve\alpha_2^\sharp
\end{pmatrix}
=
T^- \begin{pmatrix}
	\breve\alpha_1^\flat\\
	\breve\alpha_2^\flat
\end{pmatrix}.
\]
From the equation involving $T^-$, we obtain
\[
\breve \alpha_2^\sharp = t_{21}^- \breve \alpha_1^\flat + t_{22}^- \breve \alpha_2^\flat.
\]
Now using the Maslov correction at the turning points $b_0$ and $a'_0$ (see \cite[Lemma 6.1]{FMW3}), we get
\[
\begin{cases}
	\breve \alpha_1^\flat e^{\frac{i}{h}\mc{S}_1 - i\frac{\pi}{2}} &= \alpha_1^\sharp = t_{11}^+\\
	\breve \alpha_2^\flat e^{\frac{i}{h}\mc{S}_2 - i\frac{\pi}{2}} &= \alpha_2^\sharp = t_{21}^+\\
\end{cases}
\]
where 
$\mc{S}_j = \mc{S}_j(E)$ is the classical action from $\rho_+$ to $\rho_-$ on $\Gamma_j$. 
This yields, according to Theorem \ref{Thm:AsymptoticConnectionFormula},
\begin{equation}\label{Eq;breve alpha_2sharp}
	\breve \alpha_2^\sharp = t_{21}^- t_{11}^+ e^{-\frac{i}{h}\mc{S}_1 + i\frac{\pi}{2}} + t_{22}^- t_{21}^+ e^{-\frac{i}{h}\mc{S}_2 + i\frac{\pi}{2}} = \widetilde{D}(E)h^{\frac{k+1}{m+1}} + \BigO{h^{\frac{k+2}{m+1}}}
\end{equation}
where $\widetilde{D}(E) := h^{-(k+1)/(m+1)}\left( t_{21}^- e^{-i\mc{S}_1/h + i\pi/2} + t_{21}^+ e^{-i\mc{S}_2/h + i\pi/2}\right)$. 
Now, we want to compute the term $\norm{u}{L^2(x_0, +\infty)}^2$ appearing in \eqref{Eq:Im(E)ExpressedWithuandalpha}. 
We claim that
\begin{equation}\label{Eq:NormOfu}
	\norm{u}{L^2(x_0, +\infty)}^2 = 2 \mc{A}'(E) + \BigO{h^{\frac{1}{3}} + h^{\frac{k+1}{m+1}}}.
\end{equation}
where $\mc{A}(E)$ is (the analytic continuation around $E=0 \in \Cc$ of) the area of the closed curve $\Gamma_1(E)$ defined in \eqref{Eq:a_0(E)}.
To see this, we use the asymptotic behavior of the WKB solutions. The asymptotic behavior in the variable $x$ is such that $u$ decays exponentially at $x \to +\infty$. For this reason, if we fix an arbitrary $x_1 > a'_0$ then the contribution to $\norm{u}{L^2(x_0, +\infty)}^2$ mainly comes from $[x_0,x_1]$. 
On this interval, the coordinates of $u =: (u_1,u_2)^T$ behave asymptotically in $h \to 0^+$ as
\begin{flalign*}
	\norm{u_1}{L^2((x_0,x_1))}^2 &= 2\mc{A}'(E) + \BigO{h^{\frac{1}{3}}}\\
	\norm{u_2}{L^2((x_0,x_1))}^2 &= \BigO{h^{\frac{k+1}{m+1}}}
\end{flalign*}
Essentially, the term in $h^{1/3}$ appearing in $u_1$ comes from turning points 
whereas the error on $u_2$ comes from an oscillating integral and \eqref{Eq:DegenerateStationaryPhaseEstimate} of Lemma \ref{Lemma:DegenerateStationaryPhase}. We can for example refer to \cite[Proposition 7.2]{FMW3} for more details. Finally, \eqref{Eq:NormOfu} holds and we use it in \eqref{Eq:Im(E)ExpressedWithuandalpha} to get
\[
\Im(E) = -D(E)h^{1+2\frac{k+1}{m+1}} + \BigO{h^{1+2\frac{k+1}{m+1}+s}}
\]
with $s := \min\left( 1/3, 1/(m+1)\right)$ and $D(E) := \va{\widetilde{D}(E)}^2 
\left(2 \mc{A}'(E)\right)^{-1}$.
The asymptotic expansion in the case where both $k$ and $m$ are odd and $k+1 < m$ is computed in a similar way. More precisely, according to the equation \eqref{Eq:AsymptoticExpansionTransferMatrixOdd} of Theorem \ref{Thm:AsymptoticConnectionFormula}, the equation \eqref{Eq;breve alpha_2sharp} becomes
\[
\breve \alpha_2^\sharp = \widetilde{D_{odd}}(E)h^{\frac{k+2}{m+1}} + \BigO{h^{\frac{k+3}{m+1}}}.
\]
with $\widetilde{D_{odd}}(E) := h^{-\frac{k+2}{m+1}}\left( t_{21}^- e^{-\frac{i}{h}\mc{S}_1 + i\frac{\pi}{2}} + t_{21}^+ e^{-\frac{i}{h}\mc{S}_2 + i\frac{\pi}{2}}\right)$. 
We obtain 
\[
\Im(E) = -D_{odd}(E)h^{1+2\frac{k+2}{m+1}} + \BigO{h^{1+2\frac{k+2}{m+1}+s}}
\]
with $D_{odd}(E) := \va{\widetilde{D_{odd}}(E)}^2 \left(2 \mc{A}'(E)\right)^{-1}$. 

Concerning the proof of \eqref{Eq:zh(E)}, we use the concept of pseudo-resonances introduced in \cite[Section 4]{AFH}. We construct what we call a \textit{monodromy matrix} $M(E,h)$ on $\Gamma \setminus (\gamma_2^\flat \cup \breve\gamma_1^\sharp)$ and define the pseudo-resonances to be the zeros $E$ in $\mc{R}_h$ of $\det(I - M(E,h))$. Essentially, those energies $E$ satisfy a quantization condition (asymptotically in $h$) similar to the Bohr-Sommerfeld quantization rule \eqref{Eq:Bohr-Sommerfeld quantization rule} in the following sense
: for each small $h>0$, there exists a bijection $\widetilde{z}_h$ from $\mathfrak{B}_h$ to the set of pseudo-resonances satisfying
\[
\widetilde{z}_h(E) = E + \BigO{h^{1+2\frac{k+1}{m+1}}}
\]
when $k$ or $m$ is even and
\[
\widetilde{z}_h(E) = E + \BigO{h^{1+2\frac{k+2}{m+1}}}
\]
when both $k$ and $m$ are odd and $k+1 < m$. The proof of this result relies on Theorem \ref{Thm:AsymptoticConnectionFormula}. Finally, a result relying on analytic distortion and the Kato-Rellich theorem (\cite[Proposition 4.10]{AFH}) ensures that pseudo-resonances approach resonances: for any $N \in \Nn$, there exists a bijection $z_h: \mathfrak{B}_h \to \Res_h(P)$ such that
\[
\va{\widetilde{z}_h(E) - z_h(E)} = \BigO{h^N}.
\]
Combining these two results, we obtain \eqref{Eq:zh(E)}.
\end{proof}

\section{Crossing at a turning point}\label{Section:typeB}
In this section, we generalize \cite[Remark 2.4]{AFH2} and establish a connection formula (Theorem \ref{Thm:AsymptoticConnectionFormulaTypeB}) in the case \ref{Case:II}. In this case, the characteristic sets $\Gamma_j(0)$ intersect at a unique crossing point $\rho = (0,0)$ which is also a turning point, as illustrated on Figure \ref{Fig:PhaseSpaceCrossingsTypeB} below. 

\begin{figure}[h]
	{\centering
		\begin{minipage}{0.50\textwidth}
		    \centering
			\includegraphics[width=7cm]{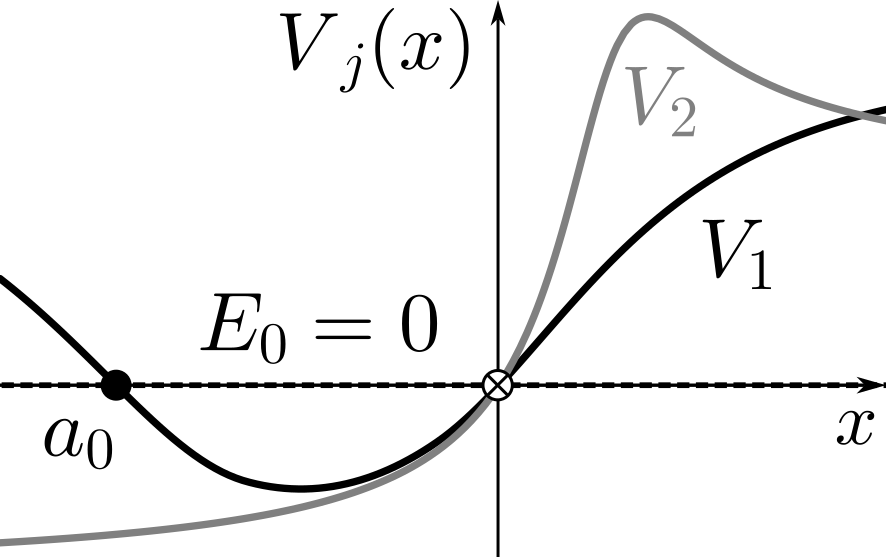}
			\caption{Potential crossing - case \ref{Case:II}}
			\label{Fig:PotentialsCrossingsTypeB}
		\end{minipage}\hfill
		\begin{minipage}{0.45\textwidth}
		    \centering
			\includegraphics[width=6.5cm]{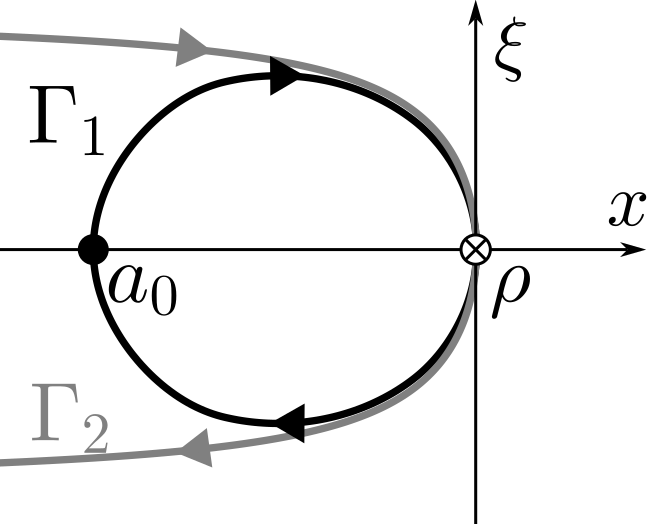}
			\caption{Associated phase space crossings - case \ref{Case:II}}
			\label{Fig:PhaseSpaceCrossingsTypeB}
		\end{minipage}
	}
\end{figure}

In this setting, the contact order of the characteristic sets $\Gamma_j(0)$ is $m = 2n$. 
Moreover, we fix $\delta_1 = \delta_1(h) := \Smallo{h^{4/(m+1)}}$ and $L > 0$ independent of $h$ and redefine 
\[
\mc{R}_h := [-\delta_1, \delta_1] + i [-hL, hL] \text{ and } \mathfrak{B}_h := \left\{E \in [ -\delta_1, \delta_1]; \; \cos\left( \frac{\mc{A}(E)}{2h} \right) = 0\right\}.
\]

\begin{thm}\label{Thm:MainResultTypeB}
Suppose that assumptions \ref{Assumption:AnalyticContinuation} up to \ref{Assumption:FiniteOrderInteraction} and case \ref{Case:II} hold true. Also assume that $m/2 \geq 2$. 
Then, for all small $h >0$ there exists a bijective map $z_h: \mathfrak{B}_h \to \Res_h(P)$ such that for any $E \in \mathfrak{B}_h$ one has 
\begin{equation}\label{Eq:zh(E)TypeB}
\va{z_h(E) - E} = O\left(h^{1+2\frac{2k+1}{m+1}}\right)
\end{equation}
and
\begin{equation}\label{Eq:Im(E)TypeB}
\Im z_h(E) = -\frac{\kappa^2}{\mc{A}'(0)}h^{1+2\frac{2k+1}{m+1}} + O\left(h^{1+2\frac{2k+1}{m+1} + \frac{1}{m+1}}\right)
\end{equation}
with 
$\kappa$ as defined in \eqref{Eq:Kappa}. 
\end{thm}

The proof of Theorem \ref{Thm:MainResultTypeB} essentially relies on Theorem \ref{Thm:AsymptoticConnectionFormulaTypeB} below. As in Section \ref{Section:ConnectionFormula}, we construct WKB solutions $f_j^\bullet$ of the equation $(P-E)f_j^\bullet \equiv 0$ near $\rho_j^\bullet$ away from $\rho$ on one hand, and exact solutions on a small interval $0 \in I$ on the other hand. 
We define
\begin{equation}\label{Eq:frhoTypeB}
f_j^\bullet(x) := 
e^{\frac{i}{h}\phi_j^\bullet(x)}\begin{pmatrix}
\sigma_{j,1}^\bullet(x,h)\\
\sigma_{j,2}^\bullet(x,h)
\end{pmatrix},
\quad
\sigma_{j,k}^\bullet(x,h) \sim \sum_{l=0}^{+\infty} h^l \sigma_{j,k,l}^\bullet(x,h)
\end{equation}
where the phases are defined for small $x < 0$ as
\[
\phi_j^\flat(x) := \int_{a_j(E)}^x\sqrt{E-V_j(y)}\,dy \stext{and} \phi_j^\sharp(x) := -\int_{a_j(E)}^x\sqrt{E-V_j(y)}\,dy
\]
and the amplitudes $\sigma_{j,k,l}^\bullet(x,h)$ satisfy
\[
\sigma_{j,j,0}^\bullet(x,h) = \left( E-V_j(y) \right)^{-1/4} \stext{and} \sigma_{j,\hat{j},0}^\bullet(x,h) = 0
\]
where $\{j, \hat{j}\} = \{1,2\}$ and $\bullet \in \{\flat, \#\}$. Here, $a_j(E)$ stands for the unique root $V_j(a_j(E)) = 0$ near $x=0$. Under Assumption \ref{Assumption:AnalyticContinuation}, $a_j(E)$ is analytic with respect to $E$, and thus so is $\phi_j^\bullet$. 

Then for any microlocal solution $w \in L^2(\Rr,\Cc^2)$ of $(P-E)w \equiv 0$ near $(0,0)$, there exist $\alpha_j^\bullet \in \Cc$ such that $w \equiv \alpha_j^\bullet f_j^\bullet \text{ near } \rho_j^\bullet$.

\begin{thm}\label{Thm:AsymptoticConnectionFormulaTypeB}
Let $E \in \mc{R}_h$. Under the same assumptions as Theorem \ref{Thm:MainResultTypeB}, the space of microlocal solutions $w \in L^2(\Rr, \Cc^2)$ of the equation $(P-E)w \equiv 0$ near the crossing point $\rho$ is of dimension $2$: there exists a $2 \times 2$ transfer matrix $T = (t_{i,j})$ depending on $E,h$ and $\rho_j^\bullet$ such that, for any microlocal solution $w \in L^2(\Rr,\Cc^2)$ of $(P-E)w \equiv 0$ near $\rho$ with $\norm{w}{L^2} \leq 1$,
\begin{equation}\label{Eq:TransferMatrixTypeB}
\begin{pmatrix}
\alpha_1^\sharp\\
\alpha_2^\sharp
\end{pmatrix}
= T(E,h)
\begin{pmatrix}
\alpha_1^\flat\\
\alpha_2^\flat
\end{pmatrix}.
\end{equation}
Moreover, 
\begin{equation}\label{Eq:AsymptoticExpansionTransferMatrixTypeB}
iT(E,h) = I_2 -i \kappa h^{\frac{2k+1}{m+1}}\begin{pmatrix}
0&1\\
1&0
\end{pmatrix}
+ \BigO{h^{\frac{2k+2}{m+1}}+h^{\frac{2k+1}{3}}}
\end{equation}
where $\kappa$ is given by \eqref{Eq:Kappa}. 
Finally, $\kappa$ 
does not depend on the point $\rho_j^\bullet$ chosen. 
\end{thm}

\begin{proof}
The method of proof is similar as the one explained in Section \ref{Subsection:ConnectionFormula}: we construct two basis of exact solutions near $x=0$, express one as a linear combination of the other, and obtain the transfer matrix from this linear combination. This is done using the integral kernels $\widetilde{K}_{j,S}$ defined above (see \cite[Section 3]{AFH2}). Similarly to Proposition \ref{Prop:IntegralOperators}, we use \eqref{Eq:DegenerateStationaryPhaseEstimate} to prove estimates on the operator norm of those kernels, and we obtain
\[
T = \begin{pmatrix}
-i & \mc{I}(E,h)\\
\mc{I}(E,h) & -i
\end{pmatrix}
+ \BigO{h^{\frac{2k+2}{m+1}}}
\]where
\begin{equation}\label{Eq:NonReducedIntegralTypeB}
\mc{I}(E,h) := -\frac{2\pi}{h^{1/3}}\int_\Rr \frac{r_0(x)}{\sqrt{\xi_1'(x)\xi_2'(x)}} \Ai(h^{-2/3}\xi_1(x)) \Ai(h^{-2/3}\xi_2(x))\, dx.
\end{equation}
Here, 
\[
\xi_j(x) :=
\begin{cases}
- \left( \frac{3}{2}\int_x^{a_j} \sqrt{E-V_j(t)}\, dt \right)^{2/3} &\stext{if} x \leq a_j,\\
\left( \frac{3}{2}\int_{a_j}^x \sqrt{V_j(t)-E}\, dt \right)^{2/3} & \stext{if} x > a_j,
\end{cases}
\]
and $\Ai(y) := (2\pi)^{-1} \int_{\Rr} \exp\left( i\left(\eta^3/3 + y\eta \right) \right) d\eta$ is the so-called \textit{Airy function}. 
Note that this function extends near $E=0 \in \Cc$ to an analytic function. The proof is then reduced to studying the behavior of the integrand of $\mc{I}(E,h)$ and using the stationary phase method. First, the asymptotic behavior of the Airy function $\Ai$ at $- \infty$ (see \cite[Chapter 11]{Olv}) is
\[
\Ai(y) = (2\pi)^{-1/2} y^{-1/4} \exp\left(-\frac{2}{3}y^{3/2}\right) \left(1 + \BigOO{y \to +\infty}{y^{-3/2}}\right)
\]
and implies that, in the region $x \geq a_1 + Ch^{2/3}$, the integrand is exponentially small. More precisely, since $\va{\lim_{x \to +\infty}(V_j(x))} > 0$, we have $\lim_{x \to +\infty} (\xi_j(x)) = +\infty$ so that the integrand is $\BigOO{x \to +\infty}{\exp\left(-g(x)/h\right)}$ with some function $g(x) \underset{x \to +\infty}{\longrightarrow} +\infty$. Integrated, this yields a term of order $\BigO{h^\infty}$. 

Next, we consider some smooth function $\chi_\nu(x)$ equal to 1 on $[2Ch^{2\nu},+\infty)$ and identically vanishing on $(-\infty, Ch^{2\nu}]$, where $C>0$ and $0 < \nu \leq 1/3$ is a parameter. In the region $|a_1-x| \leq Ch^{2/3}$, the integrand is $\BigO{h^{-1/3}r_0(x)}$. In the region $2Ch^{2/3} \leq a_1 - x \leq Ch^{2\nu}$, the integrand is $\BigO{r_0(x)(x-a_1)^{-1/2}}$. To prove this, we used the fact that $\xi_j(x) = \xi_j'(a_1)(x-a_1) + \BigOO{x \to a_1}{(x-a_1)^2}$ with $\xi_j'(a_1) \neq 0$, as well as the asymptotic behavior
\begin{equation}\label{Eq:AymptoticBehaviorAiryFunction-inf}
\Ai(-y) = \pi^{-1/2} y^{-1/4} \sin \left(\frac{2}{3}y^{3/2} + \frac{\pi}{4} \right) + \BigOO{y \to +\infty}{y^{-7/4}}
\end{equation}
of $\Ai$ at $- \infty$. This allows us to show that
\begin{equation}\label{Eq:SmallContributionIntegralTypeB}
-\frac{2\pi}{h^{1/3}}\int_\Rr (1-\chi_\nu(a_1-x))\frac{r_0(x)}{\sqrt{\xi_1'(x)\xi_2'(x)}} \Ai(h^{-2/3}\xi_1(x)) \Ai(h^{-2/3}\xi_2(x))\, dx = \BigO{h^{(2k+1)\nu}}.
\end{equation}
This implies that the main contribution to \eqref{Eq:NonReducedIntegralTypeB} is given by
\begin{equation}\label{Eq:ReducedIntegralTypeB}
    \mc{I}_0(E,h) := -\sum_{\pm} \int_{-\delta}^\delta \chi_\nu(a-x)\sigma(x,E)e^{\pm\frac{i}{h}\phi(x,E)}\, dx.
\end{equation}
where $a := \min(a_1,a_2) = a_j + \BigO{E^{m/2}}$ and $\delta >0$ is a small parameter. Here the amplitude $\sigma(x,E)$ is defined by
\[
\sigma(x,E) := \frac{1}{2}r_0(x)\left[(E-V_1(x))(E-V_2(x))\right]^{-1/4}
\]
and the phase $\phi(x,E)$ is defined by
\begin{equation}\label{Eq:PhaseForComputationOfKappa}
\phi(x,E) := \int_{a_1(E)}^x \sqrt{E-V_1(t)}\, dt - \int_{a_2(E)}^x \sqrt{E-V_2(t)}\, dt.
\end{equation}
Following \cite[Section 4]{AFH2}, we fix $\nu := \min\{1/3,2/(m+1)\}$ and $\nu' :=(m-1)/(m(m+1))$ and split the study of $\mc{I}_0$ in two. 
On one hand we obtain, via an integration by parts,
\[
-\sum_{\pm} \int_{-\delta}^\delta (\chi_\nu\chi_{\nu'})(a-x)\sigma(x,E)e^{\pm\frac{i}{h}\phi(x,E)}\, dx = \BigO{h^{1+\nu'(2k-m)}} = \BigO{h^{(2k+1)\nu}}
\]
On the other hand, one can show that there exists a smooth function $f(y;E)$ defined near $y=0$ satisfying $f(0;E) = 1 + \BigO{E}$ such that the asymptotic behavior of the above functions for $h \to 0^+$ and $x \in \supp\left([\chi_\nu (1-\chi_{\nu'})](a-\cdot)\right)$ is given by
\[
\begin{cases}
	\sigma(x,E) &= \frac12 \frac{r_0^{(k)}(0)}{k!}\frac{(a-x)^{k-1/2}}{V_1'(0)^{1/2}}\left(1 + \BigO{h^{2\nu'}}\right)\\
	\phi(x,E) &= \frac{(-1)^{m/2}q_m}{m+1}f(a-x;E)(a-x)^{(m+1)/2} + \BigO{h^{1+1/m}}
\end{cases}
\]
where $q_m := v_{m/2}\left((m/2)!\right)^{-1}V_1'(0)^{-1/2}.$
Here, we used the fact that $E = \Smallo{h^{2\nu}}$ on the support of $[\chi_\nu (1-\chi_{\nu'})](a-\cdot)$. We compute this oscillating integral using \eqref{Eq:DegenerateStationaryPhase3Cases}. We finally get $\mc{I}_0(E,h) = \kappa h^{\frac{2k+1}{m+1}} + \BigO{h^{\frac{2k+2}{m+1}}+h^{\frac{2k+1}{3}}}$ where
\begin{equation}\label{Eq:Kappa}
\kappa := \frac{2r_0^{(k)}(0)}{k!V_1'(0)^{1/2}}\cos\left(\frac{2k+1}{2(m+1)}\pi\right)\pmb\Gamma\left(\frac{2k+1}{m+1}\right)  \left(\frac{m+1}{q_m}\right)^{\frac{2k+1}{m+1}}.
\end{equation}

\end{proof}

\section{Generalization to an interaction $U=r_0(x) + ihr_1(x)D_x$}\label{Section:r1}
We generalize the results of this article in the case where the interaction term is a first order differential operator with smooth coefficients, that is to say of the form
\[
U = r_0(x) + ihr_1(x)D_x.
\]
with smooth $r_0$ and $r_1$. This is a similar situation as the framework of \cite{FMW2} (where $r_0$ vanishes identically). 
We deal with \ref{Case:Ia} and \ref{Case:Ib} in Proposition \ref{Prop:MainResultr1TypeA}, and with \ref{Case:II} in Proposition \ref{Prop:MainResultr1TypeB}.

\begin{prop}\label{Prop:MainResultr1TypeA}
Suppose that assumptions \ref{Assumption:AnalyticContinuation} up to \ref{Assumption:FiniteOrderForPotentials} and either \ref{Case:Ia} or \ref{Case:Ib} hold true. Assume also that $r_0$ (resp. $r_1$) satisfies Assumption \ref{Assumption:FiniteOrderInteraction} and write $k_0$ (resp. $k_1$) the associated vanishing order. 
Then, in case \ref{Case:Ia}, for all small $h >0$ there exists a bijective map $z_h: \mathfrak{B}_h \to \Res_h(P)$ such that for any $E \in \mathfrak{B}_h$, \eqref{Eq:zh(E)}, \eqref{Eq:Im(E)} and \eqref{Eq:D(E)} hold true with $k =\min(k_0,k_1) < m$. In case \ref{Case:Ib} and with , for all small $h >0$ there exists a bijective map $z_h: \mathfrak{B}_h \to \Res_h(P)$ such that for any $E \in \mathfrak{B}_h$, \eqref{Eq:zh(E)Odd}, \eqref{Eq:Im(E)Odd} and \eqref{Eq:D(E)Odd} hold true with $k =\min(k_0,k_1) < m+1$. 
Moreover, in case \ref{Case:Ia} the $\omega$ of\eqref{Eq:Omega} becomes
\[
\omega = \frac{\mu_{k,m}\left(\frac{\sgn(v_m)(k+1)}{2(m+1)}\pi\right)}{(m+1)k!} 
\pmb\Gamma\left(\frac{k+1}{m+1}\right)  E_0^{\frac{k-m}{2(m+1)}} \left(\frac{2(m+1)!}{\va{v_m}}\right)^{\frac{k+1}{m+1}}\partial_x^k \overline{U}(\rho_+).
\]
and in case \ref{Case:Ib} the $\omega_{odd}$ is obtained replacing the term $\tau r^{(k)}(0) + r^{(k+1)}(0)$ in \eqref{Eq:OmegaOdd} by
\begin{equation}\label{Eq:nur_1}
\tau \partial_x^{k} \overline{U}(\rho_+) + \partial_x^{k+1} \overline{U}(\rho_+) +i V_1'(0) \frac{(k+1)}{2\sqrt{E_0}} r_1^{(k)}(0).
\end{equation}
\end{prop}

\begin{proof}
We redefine the integral operators $\widetilde{K}_{j,S}$ of Section \ref{Section:ConnectionFormula} by
\[
\widetilde{K}_{j,S}f(x) := \frac i2\left(\breve u_j(x)\int_{\pm\delta}^x u_j(y)(U_jf)(y)dy - u_j(x)\int_{\pm\delta}^x \breve u_j(y)(U_jf)(y)dy\right)
\]
where $U_1 := U$ and $U_2 := U^*$, $j \in \{1,2\}$ and $S \in \{L,R\}$. Using an integration by parts we remark that $\widetilde{K}_{j,S}$ is still defined as a bounded linear operator on $\C(I)$. Here, $U_2$ has a symbol of $r_0(x) - ir_1(x)\xi$ and therefore integrals involving $U_1$ and those involving $U_2$ can be treated in a symmetric way. The proof of Proposition \ref{Prop:IntegralOperators} is conducted in a similar way as in the case where $U = r_0(x)$. Indeed, when expanding $\widetilde{K}_{j,S}\widetilde{K}_{\hat{j},S}f(x)$ we obtain a sum of terms of order $h$, of integrals of the same type as the case $U=r_0(x)$ and of integrals of the form
\[
\int_{-\delta}^x u_1(y) r_1(y)F(y) \breve u_2(y)\phi_2'(y) dy
\]
where $F(z) = \int_{-\delta}^y u_2(z) (U_1 f)(z) dz$ is defined on $\C(I)$ via integration by parts. We then apply \eqref{Eq:DegenerateStationaryPhaseEstimate} of Lemma \ref{Lemma:DegenerateStationaryPhase} with $a(y) = a_1(y) r_1(y)F(y) \breve a_2(y)\phi_2'(y)$ and $\phi := \phi_2 - \phi_1$. Here, $a_j$ and $\phi_j$ are the amplitudes and phases corresponding to $u_j$ and $\breve u_j$ in \eqref{Eq:AsymptoticBehaviorScalarSolutions}.
The computation of $\omega$ and $\omega_{odd}$ in the proof of Lemma \ref{Lemma:AlphaBetaEstimates} is also very similar: we conduct calculations using the fact that
\begin{equation}\label{Eq:IntegralForOmegar1TypeA}
\int_{-\delta}^{\delta} (\breve u_2U_2 u_1)(y)dy = \int_{-\delta}^{\delta} \breve u_2(y) r_0(y) u_1(y)dy - i\int_{-\delta}^{\delta} \breve u_2(y) r_1(y) \phi_1'(y) u_1(y)dy + \BigO{h}.
\end{equation}
To prove \eqref{Eq:nur_1}, we use
\begin{flalign*}
(r_0-ir_1\phi_1')^{(k+1)}(0) = \partial_x^{k+1} \overline{U}(\rho_+) +i V_1'(0) \frac{(k+1)}{2\sqrt{E_0}} r_1^{(k)}(0) + \BigO{h}.
\end{flalign*}
Notice that, in the odd case, $1 \leq k+1 \leq m-1$ so that $V_1'(0) = V_2'(0)$. 
\end{proof}

\begin{rmk}
In the cases \ref{Case:Ia} and \ref{Case:Ib}, the role played by $r_0$ and $r_1$ is the same. Intuitively, this comes from the fact that the symbol of the interaction, $r_0(x) + ihr_1(x)\xi$, vanishes in $x$ at the same order as $\min(k_0,k_1)$. In the case \ref{Case:II} however, the symbol $ihr_1(x)\xi$ vanishes at order $1$ in $\xi$, which is why we can expect $r_1$ to have a less important role compared to $r_0$, as described below.
\end{rmk}

\begin{prop}\label{Prop:MainResultr1TypeB}
Suppose that assumptions \ref{Assumption:AnalyticContinuation} up to \ref{Assumption:FiniteOrderForPotentials} as well as \ref{Case:II} hold true. Assume also that $r_0$ (resp. $r_1$) satisfies Assumption \ref{Assumption:FiniteOrderInteraction} and write $k_0$ (resp. $k_1$) the associated vanishing order. 
Then, for all small $h >0$ there exists a bijective map $z_h: \mathfrak{B}_h \to \Res_h(P)$ such that for any $E \in \mathfrak{B}_h$  \eqref{Eq:zh(E)TypeB} and \eqref{Eq:Im(E)TypeB} hold true with $k :=\min(k_0,k_1+1/2)$ and $\kappa$ given either by \eqref{Eq:Kappa} when $k=k_0$ or by
\begin{equation}\label{Eq:Kappar1}
\kappa_1 := \frac{2r_1^{(k_1)}(0)}{k_1!}\cos\left(\frac{2k_1+2}{2(m+1)}\pi\right)\pmb\Gamma\left(\frac{2k_1+2}{m+1}\right)  \left(\frac{m+1}{q_m}\right)^{\frac{2k_1+2}{m+1}}
\end{equation}
when $k=k_1+1/2$.
\end{prop}

\begin{proof}
The constant $\kappa$ appearing in the sub-principal term of the transfer matrix is given by the principal term in the $h$-asymptotic expansion of the integral
\[
-\frac{2\pi}{h^{1/3}}\int_{-\delta}^\delta \frac{1}{\sqrt{\xi_2'(x)}}  \Ai(h^{-2/3}\xi_2(x)) \times (r_0+ihr_1(x)D_x) \left( \frac{1}{\sqrt{\xi_1'(x)}}\Ai(h^{-2/3}\xi_1(x)) \right) dx.
\]
Since we already dealt with the term containing $r_0$, it suffices to deal with integrals of the form
\begin{equation}\label{Eq:NonReducedIntegralTypeBr1}
-2\pi h^{2/3} \int_{-\delta}^\delta \frac{r_1(x)}{\sqrt{\xi_2'(x)}}  \Ai(h^{-2/3}\xi_2(x)) \partial_x \left( \frac{1}{\sqrt{\xi_1'(x)}}\Ai(h^{-2/3}\xi_1(x)) \right) dx.
\end{equation}
In a similar fashion as \eqref{Eq:SmallContributionIntegralTypeB}, we claim that the above integral cut off with $1 - \chi_\nu$ is of order $h^{(2k_1+2)\nu}$
for $\nu$ and $\chi_\nu$ as in Section \ref{Section:typeB}. Here, we need to take in account the asymptotic behavior
\begin{equation}\label{Eq:AymptoticBehaviorDerivativeOfAiryFunction-inf}
\Ai'(-y) = \pi^{-1/2} y^{1/4} \sin \left(\frac{2}{3}y^{3/2} - \frac{\pi}{4} \right) + \BigOO{y \to +\infty}{y^{-5/4}}
\end{equation}
of the derivative $\Ai'$ of the Airy function (see \cite[Chapter 11]{Olv}). In particular, $\Ai(-y)\Ai'(-y)$ is bounded at $y \to +\infty$. With this, we can prove that the two terms
\[
-2\pi h^{2/3}\int_{-\delta}^\delta (1-\chi_\nu(a_1-x))\frac{r_1(x)}{\sqrt{\xi_2'(x)}}  (\xi_1'(x))^{-3/2}\xi_1''(x)\Ai(h^{-2/3}\xi_1(x)) \Ai(h^{-2/3}\xi_2(x)) \, dx
\]
and
\[
-2\pi \int_{-\delta}^\delta (1-\chi_\nu(a_1-x))\frac{r_1(x)}{\sqrt{\xi_2'(x)}} (\xi_1'(x))^{1/2} \Ai'(h^{-2/3}\xi_1(x)) \Ai(h^{-2/3}\xi_2(x)) \, dx
\]
of the integral \eqref{Eq:NonReducedIntegralTypeBr1} cut off with $1 - \chi_\nu$ are indeed $\BigO{h^{(2k_1+2)\nu}}$. Moreover, the integral corresponding to \eqref{Eq:NonReducedIntegralTypeBr1} cut off with $\chi_\nu$ has its first term
\[
-2\pi h^{2/3}\int_{-\delta}^\delta \chi_\nu(a_1-x)\frac{r_1(x)}{\sqrt{\xi_2'(x)}}  (\xi_1'(x))^{-3/2}\xi_1''(x)\Ai(h^{-2/3}\xi_1(x)) \Ai(h^{-2/3}\xi_2(x)) \, dx
\]
of order $h$ because of the asymptotic behavior \eqref{Eq:AymptoticBehaviorAiryFunction-inf} of $\Ai$ and because $\xi_1''(x) = \BigO{(a_1-x)^{-1/2}}$. The second term
\[
-2\pi \int_{-\delta}^\delta \chi_\nu(a_1-x)\frac{r_1(x)}{\sqrt{\xi_2'(x)}} (\xi_1'(x))^{1/2} \Ai(h^{-2/3}\xi_2(x)) \Ai'(h^{-2/3}\xi_1(x)) \, dx.
\]
is an oscillatory integral of the form $\int_{-\delta}^\delta \chi_\nu(a_1-x)\sigma(x) e^{\frac{i}{h}\phi(x)} \, dx$ with $\phi$ as in \eqref{Eq:PhaseForComputationOfKappa} and
\[
\sigma(x) :=  \frac{r_1^{(k_1)}(0)}{k!}\frac{(a-x)^{k_1}}{V_1'(0)^{1/2}}\left(1 + \BigO{h^{2\nu'}}\right).
\]

As in Section \ref{Section:typeB} we compute this integral using an appropriate cutoff and \eqref{Eq:DegenerateStationaryPhase3Cases} and get \\$\int_{-\delta}^\delta \sigma(x) e^{\frac{i}{h}\phi(x)} \, dx = \kappa_1 h^{\frac{2k_1+2}{m+1}} + \BigO{h^{\frac{2k+3}{m+1}}+h^{\frac{2k+2}{3}}}$ with $\kappa_1$ as in \eqref{Eq:Kappar1}.
\end{proof}

Therefore, the transfer matrix \eqref{Eq:TransferMatrixTypeB} in Theorem \ref{Thm:AsymptoticConnectionFormulaTypeB} has a sub-principal term of order $h^{\frac{2k+1}{m+1}}$, with $k = \min(k_0,k_1+1/2)$, as claimed. Remark that $k_0$ is never equal to $k_1+1/2$, hence the coefficient of that sub-principal term is given either by the $\kappa$ defined in \eqref{Eq:Kappa} or by the $\kappa_1$ defined in \eqref{Eq:Kappar1}. 

\begin{rmk}
The fact that $\Ai'$ appears instead of $\Ai$ makes the amplitude $\sigma$ vanish at order $k_1$, instead of $k_0 - 1/2$ in \eqref{Eq:ReducedIntegralTypeB}. This explains the $+1/2$ appearing in $k = \min(k_0,k_1+1/2)$. Another way to see this is to remark that, if we note $\sigma(x)e^{\frac{i}{h}(\phi_1(x)-\phi_2(x))}$ the oscillatory function appearing in the integral \eqref{Eq:ReducedIntegralTypeB} (case $r_1 = 0$), then the integrand of the integral giving $\kappa_1$ is  essentially $\sigma(x)\phi_1'(x)e^{\frac{i}{h}(\phi_1(x)-\phi_2(x))}$. In the Section \ref{Section:ConnectionFormula} (where the derivative of the phase was the same), this factor $\phi_1'(x) = \sqrt{E - V_1(x)}$ in \eqref{Eq:IntegralForOmegar1TypeA} did not vanish at $x=0$ because the crossing point was not a turning point. However in Section \ref{Section:typeB}, the crossing point is a turning point therefore $\phi_1'(x)$ vanishes at order $1/2$ in $x$.
\end{rmk}

Note that the microlocal connection formula (Theorem \ref{Thm:AsymptoticConnectionFormula}) is improved in a similar way, taking the $k$ and $\omega$ (or $\omega_{odd}$ for the case \ref{Case:Ib}, and $\kappa_1$ for the case \ref{Case:II}) defined above.

\section*{Acknowledgement}
I would like to warmly thank Kouichi Taira and Kenta Higuchi whose remarks and advices were most valuable during the writing of this article.

\appendix
\section{A degenerate stationary phase estimate}\label{Appendix:DegenerateStationaryPhaseEstimates}
Our goal in this section is to study the asymptotic behavior in the semiclassical limit $h \to 0^+$ of integrals of the form
\[
\int e^{\frac{i}{h} \phi(y)}\sigma(y)dy.
\]
In this article, there are two situations that require two different types of results on this oscillatory integral.
\begin{enumerate}
    \item[(i)] For integrals of the form $\int_\alpha^{\beta} e^{\frac{i}{h} \phi(y)}\sigma(y)dy$ with fixed $\alpha$ and $\beta$, we can use a cutoff near $0$ and consider it as an integral of a smooth function compactly supported near $0$. 
    \item[(ii)] For integrals of the form $\int_\alpha^x e^{\frac{i}{h} \phi(y)}\sigma(y)dy$ with fixed $\alpha$ and variable $x$, we must take into account the boundary term at $y=x$, and find a bound that does not depend on $x$. In this situation, we need to assume that $k \leq m$. 
\end{enumerate}
The following lemma accounts for both of these cases.
\begin{lemma}\label{Lemma:DegenerateStationaryPhase}
Let $I$ be a bounded open interval of $\Rr$ containing $0$ and fix $\alpha \in I$. 
Let $\sigma  \in \C_b^\infty(I,\Cc)$ and $\phi \in \C_b^\infty(I,\Rr)$ two smooth functions with bounded derivatives at any order. We assume the following :
\begin{enumerate}%
    \item[\textbf{(H1)}] The function $\phi'$ does not vanish on $I \setminus \{0\}$ and is of the form $\phi'(y) = y^{m}\phi_1(y)$ where $\phi_1$ is a smooth function such that $\phi_1(0) \neq 0$.
    \item[\textbf{(H2)}] The function $\sigma$ is of the form $\sigma(y) = y^k \sigma_0(y)$ where $\sigma_0$ is a smooth function.
\end{enumerate}
If $k \leq m$ then there exist a constant $C>0$ independent of $h$ such that, for any $x \in I$ and any $h > 0$ small enough,
\begin{equation}\label{Eq:DegenerateStationaryPhaseEstimate}
    \va{\int_\alpha^x e^{\frac{i}{h} \phi(y)}\sigma(y)dy} \leq C \left(h^{\frac{k+1}{m+1}} \norm{\sigma_0}{L^\infty(I)} 
    + h^{\frac{k+2}{m+1}} \norm{\sigma_0'}{L^\infty(I)}\log\left( \frac{1}{h}\right)^{\delta_{m,k+1}}\right)
\end{equation}
\noindent
where $\delta_{k,m+1}$ is the Kronecker symbol. Furthermore, regardless of $k$ and $m$, the following asymptotic expansion holds: 
\begin{equation}\label{Eq:DegenerateStationaryPhase3Cases}
e^{-\frac{i}{h}\phi(0)} \int_I \sigma(y)e^{\frac{i}{h}\phi(y)}dy = 
\frac{2\mu_{k,m}\left(\frac{\varepsilon(\phi)(k+1)}{2(m+1)}\pi\right)}{m+1}
\pmb\Gamma\left(\frac{k+1}{m+1}\right) 
\left(\frac{m+1}{\va{\phi_1(0)}}\right)^{\frac{k+1}{m+1}}
\sigma_0(0)h^{\frac{k+1}{m+1}}\\
+ \BigO{h^{\frac{k+2}{m+1}}}
\end{equation}
\noindent
where $\pmb\Gamma$ is the Gamma function, $\varepsilon(\phi)$ is $\sgn(\phi_1(0))$ and $\mu_{k,m}$ is defined in \eqref{Eq:mu_kmtheta} replacing $v_m$ by $\phi_1(0)$. When $k$ and $m$ are both odd the value of the previous integral is, with error $\BigO{h^{\frac{k+3}{m+1}}}$,
\begin{equation}\label{Eq:DegenerateStationaryPhaseOddCase}
\frac{2\mu_{k+1,m}\left(\frac{\varepsilon(\phi)(k+2)}{2(m+1)}\pi\right)}{m+1} 
\pmb\Gamma\left(\frac{k+2}{m+1}\right)
\left(\frac{m+1}{\va{\phi_1(0)}}\right)^{\frac{k+2}{m+1}}
\left[
\sigma_0'(0) - \frac{2(2k+1)\phi_1'(0)}{(m+2)\va{\phi_1(0)}}\sigma_0(0)
\right] h^{\frac{k+2}{m+1}}.
\end{equation}
\end{lemma}
\begin{proof}  We may assume that $\sigma$ has a compact support in $I$ (use the non-stationary phase method to deal with points away from $0$).
Let $\rho \in \C^\infty_c(\Rr,[0,1])$ be a cutoff function vanishing for $|y| \geq 2$ and equal to $1$ for $|y| \leq 1$. We set $\varepsilon:= h^{\frac{1}{m+1}}$ and we split $\mc{I}(x):= \int_\alpha^x e^{\frac{i}{h} \phi(y)}\sigma(y)dy$ into
\[
\mc{I}(x)= \mc{I}_1(x) + \mc{I}_2(x):= \int_\alpha^x e^{\frac{i}{h} \phi(y)} \sigma(y) \rho\left( \frac{y}{\varepsilon} \right) dy + \int_\alpha^x e^{\frac{i}{h} \phi(y)}\sigma(y) \left( 1 - \rho\left( \frac{y}{\varepsilon} \right) \right) dy.
\]
For $\mc{I}_1(x)$, we remark that, for $h$ small enough so that $[-2\varepsilon, 2\varepsilon] \subset I$,
\[
|\mc{I}_1(x)| \leq \norm{\sigma_0(y)}{L^\infty(I)} \int_{-2\varepsilon}^{2\varepsilon} |y|^k dy \leq \frac{2^{k+2}}{k+1} h^{\frac{k+1}{m+1}} \norm{\sigma_0(y)}{L^\infty(I)}
\]
To study the second term, we integrate by parts using the formal adjoint $L^Tf(y) := - \partial_y \left( \frac{h}{i \phi'}f\right)(y)$ of $Lf(y) := \frac{h}{i \phi'(y)} \partial_y f(y)$. Since $L\left(e^{\frac{i}{h}\phi}\right) = e^{\frac{i}{h}\phi}$, we get
\begin{equation}\label{Eq:ExpressionI2}
\mc{I}_2(x) = \int_\alpha^x e^{\frac{i}{h} \phi(y)} L^T \left[\sigma(y) \left( 1 - \rho\left( \frac{y}{\varepsilon} \right) \right) \right] dy + \frac{h}{i \phi'(x)}e^{\frac{i}{h} \phi(x)}\sigma(x) \left( 1 - \rho\left( \frac{x}{\varepsilon} \right) \right).
\end{equation}
Next, \textbf{(H1)} implies $\frac{1}{\phi'} \underset{y \to 0}{\sim} y^{-m}$ and $\frac{\phi''}{(\phi')^2} \underset{y \to 0}{\sim} y^{-(m+1)}$ so that the left term in \eqref{Eq:ExpressionI2} is bounded by
\[
C' h \int_{\{y \in I, |y| \geq \varepsilon\}} 
|y|^{k-m}\norm{\sigma_0'(y)}{L^\infty(I)} 
+ |y|^{k-m} |y|^{-1} \norm{\sigma_0(y)}{L^\infty(I)} 
dy
\]
where $C'>0$. The boundary term in \eqref{Eq:ExpressionI2}  is equal to zero when $|x| \leq \varepsilon$. When $|x| \geq \varepsilon$, \textbf{(H2)} implies that $\va{\frac{1}{\phi'(x)}} \lesssim \varepsilon^{-m}= h^{\frac{-m}{m+1}}$ 
so that there exists a constant $C''>0$ such that
\[
\va{\frac{h}{i \phi'(x)}e^{\frac{i}{h} \phi(x)} \sigma(x) \left( 1 - \rho\left( \frac{x}{\varepsilon} \right) \right)} \leq C'' h \varepsilon^{k-m}  \norm{\sigma_0(y)}{L^\infty(I)} = C'' h^{\frac{k+1}{m+1}}\norm{\sigma_0(y)}{L^\infty(I)}.
\]
Here, $k \leq m$ so that $h = \BigO{h^{\frac{k+1}{m+1}}}$. Combining the estimates previously obtained, we obtain the claimed estimate for $\mc{I}$. In particular, the $\log(1/h)$ term only appears when $k-m = -1$.


We now prove \eqref{Eq:DegenerateStationaryPhase3Cases} and \eqref{Eq:DegenerateStationaryPhaseOddCase}.
We bring the study on $\Rr$ instead of $I$ using a cutoff near $0$. The contribution away from $0$ is dealt with using the non-stationary phase method. First, we prove with the residue theorem that
\[
\int_\Rr y^k e^{\frac{i}{h}y^{m+1}}dy = \frac{2}{m+1}\pmb\Gamma\left(\frac{k+1}{m+1}\right) \mu_{k,m}\left(\frac{k+1}{2(m+1)\pi}\right)h^{\frac{k+1}{m+1}}.
\]
This is proven in a similar way to the computation of the Fresnel integral $\int_\Rr e^{-it^2}dt$. Next, we use the variable change $\psi(y) = y\left(\sgn(\phi_1(0))\frac{\phi_1(y)}{m+1}\right)^{\frac{1}{m+1}}$ and obtain
\[
\int_\Rr \sigma(y) e^{\frac{i}{h}\phi(y)}dy = \int_\Rr \sigma(\psi^{-1}(y))e^{\frac{i}{h}\sgn(\phi_1(0)) y^{m+1}} \va{\left(\psi^{-1}\right)'(y)}dy.
\]
Here \textbf{(H1)} ensures that, locally around $0$, $\left(\psi^{-1}\right)'$ is of constant sign. A straightforward computation using Taylor expansions at $0$ then gives \eqref{Eq:DegenerateStationaryPhase3Cases} and \eqref{Eq:DegenerateStationaryPhaseOddCase} as claimed.

\end{proof}


\end{spacing}

\end{document}